\documentclass[a4paper,11pt]{article}

\usepackage[margin=1in]{geometry}
\usepackage{booktabs} %
\usepackage{multirow}
\usepackage{hyperref}
\usepackage{breakurl}
\newcommand{\email}[1]{\href{mailto:#1}{\nolinkurl{#1}}}

\newcommand{\Bin}{\mathsf{Bin}}
\newcommand{\Poi}{\mathsf{Poi}}

\usepackage[round]{natbib}

\usepackage{amsmath,amsthm,amsfonts}
\usepackage{thmtools}
\usepackage{thm-restate}
\usepackage{ifthen}
\usepackage{xcolor}
\usepackage{tikz}
\usetikzlibrary{calc,intersections}
\usepackage{xspace}
\usepackage{subcaption}
\usepackage{multirow}
\usepackage{pgfplots}

\theoremstyle{plain}
\newtheorem{theorem}{Theorem}
\newtheorem{lemma}{Lemma}

\theoremstyle{definition}
\newtheorem*{example}{Example}

\renewcommand{\vec}[1]{\mathbf{#1}}
\newcommand{\E}{\mathbf{E}}
\renewcommand{\P}{\mathbf{Pr}}
\newcommand{\R}{\mathbb{R}}
\newcommand{\N}{\mathbb{N}}
\newcommand{\chr}[1]{\ifthenelse{\equal{#1}{}}{\chi}{\chi[#1]}}
\newcommand{\Bern}{\mathfrak{B}}

\newcommand{\Q}{\mathcal{Q}}
\newcommand{\F}{\hat{F}}
\newcommand{\T}{\mathcal{T}}

\newcommand{\Rs}{R^{\mathrm{s}}}
\newcommand{\Rd}{R^{\mathrm{d}}}
\newcommand{\Rm}{R^{\mathrm{m}}}
\newcommand{\Rx}{R^{\mathrm{x}}}
\newcommand{\Rl}{R^{\mathrm{l}}}
\newcommand{\tR}{\tilde{R}}
\makeatletter
\newcommand{\raisemath}[1]{\mathpalette{\raisem@th{#1}}}
\newcommand{\raisem@th}[3]{\raisebox{#1}[0pt]{$#2#3$}}
\makeatother
\newcommand{\Rsn}[1]{\Rs_{\raisemath{0pt}{F,n,#1}}}
\newcommand{\Rdn}[1]{\Rd_{\raisemath{0pt}{F,n,#1}}}
\newcommand{\Rmn}[1]{\Rm_{\raisemath{0pt}{F,n,#1}}}
\newcommand{\Rxn}[1]{\Rx_{\raisemath{0pt}{F,n,#1}}}
\newcommand{\Rsnk}{\Rsn{k}}
\newcommand{\Rdnk}{\Rdn{k}}
\newcommand{\Rmnk}{\Rmn{k}}
\newcommand{\Rxnk}{\Rxn{k}}

\newcommand{\midd}{:}
\renewcommand{\phi}{\varphi}

\newcommand{\secref}[1]{Section~\ref{#1}}
\newcommand{\thmref}[1]{Theorem~\ref{#1}}
\newcommand{\lemref}[1]{Lemma~\ref{#1}}
\newcommand{\corref}[1]{Corollary~\ref{#1}}
\newcommand{\figref}[1]{Figure~\ref{#1}}
\newcommand{\appref}[1]{Appendix~\ref{#1}}

\newcommand{\ie}{i.e.,\xspace}
\newcommand{\eg}{e.g.,\xspace}

\def\clap#1{\hbox to 0pt{\hss#1\hss}}
  \def\mathclap{\mathpalette\mathclapinternal}

\def\mathclapinternal#1#2{\clap{$\mathsurround=0pt#1{#2}$}}

\title{Revenue Gaps for Static and Dynamic Posted Pricing\\of Homogeneous Goods\thanks{Helpful discussions with Rad Niazadeh and Frieder Smolny are gratefully acknowledged. The second and third author were supported by the Einstein Foundation Berlin.}}

\author{%
	Paul D\"utting\thanks{Department of Mathematics, London School of Economics, Houghton Street, London WC2A 2AE, UK. Email: \email{p.d.duetting@lse.ac.uk}}
	\and
	Felix Fischer\thanks{School of Mathematical Sciences, Queen Mary University of London, Mile End Road, London E1 4NS, UK. Email: \email{felix.fischer@qmul.ac.uk}}
	\and
	Max Klimm\thanks{Wirtschaftswissenschaftliche Fakult\"at, Humboldt-Universit\"at zu Berlin, Spandauer Stra\ss{}e 1, 10178 Berlin, Germany, Email: \email{max.klimm@hu-berlin.de}}
}

\date{}

\begin{document}

\maketitle

\begin{abstract}
We consider the problem of maximizing the expected revenue from selling~$k$ homogeneous goods to~$n$ unit-demand buyers who arrive sequentially with independent and identically distributed valuations. In this setting the optimal posted prices are dynamic in the sense that they depend on the remaining numbers of goods and buyers. We investigate how much revenue is lost when a single static price is used instead for all buyers and goods, and prove upper bounds on the ratio between the maximum revenue from dynamic prices and that from static prices. These bounds are tight for all values of~$k$ and~$n$ and vary depending on a regularity property of the underlying distribution. For general distributions we obtain a ratio of $2-k/n$, for regular distributions a ratio that increases in $n$ and is bounded from above by $1/(1-\frac{k^k}{e^{k}k!})$, which is roughly $1/(1- \frac{1}{\sqrt{2\pi k}})$. The lower bounds hold for the revenue gap between dynamic and static prices. The upper bounds are obtained via an ex-ante relaxation of the revenue maximization problem, as a consequence the tight bounds of $2-k/n$ in the general case and of $1/(1- \frac{1}{\sqrt{2 \pi k}})$ in the regular case apply also to the potentially larger revenue gap between the optimal incentive-compatible mechanism and the optimal static price. Our results imply that for regular distributions the benefit of dynamic prices vanishes while for non-regular distributions dynamic prices may achieve up to twice the revenue of static prices.
\end{abstract}

\section{Introduction}

In the 1820s Irishman Alexander Turney Stewart opened a dry goods store in New York City and adopted a radical new policy. Rather than negotiating prices individually with each customer, he offered his goods at fixed prices. While Stewart was raised by Quaker grandparents and arrived in New York with their introduction letters to American Quakers~\citep{Snodgrass2011}, we can only speculate whether his decision was based on the Quaker ethic of honesty and non-preferential treatment as expressed in George Fox's tract that warned the merchants of London against using the ``cozening and cheating and defrauding'' practice of haggling~\citep{Fox1658,Kent1983}. Rival retailers predicted that Stewart would be bankrupt within a week. Instead he became a multi-millionaire, and the A.~T.~Stewart Dry Goods Store the largest department store in the world~\citep{Resseguie1965}. More than 150 years later, fixed prices are the main mechanism for the sale of goods and an important mechanism for the sale of services. For instance, around 50\% of the total consumer expenditure in the US is spent on goods or services with fixed prices, with vehicles and housing as notable exceptions \citep{Philipps2012}.

Fixed-price policies came under scrutiny following the Airline Deregulation Act of 1978, which removed government control from air transport routes and fares and allowed low-cost carriers like People Express Airlines to take large profits. Established carriers like the British Overseas Airways Cooperation and
American Airlines reacted by adopting a new type of policy of capacity-controlled fares where prices depend on the remaining number of seats and the time until departure \citep[e.g.,][]{McAfeeVelde06}.
Today dynamic pricing and revenue management are considered major success stories in operations research \citep[e.g.,][]{TallurivanRyzin04}, and more and more industries adopt ever more sophisticated algorithmic pricing strategies. While originally it was limited mainly to electronic commerce, dynamic pricing is set to gain traction also in physical retailing, for example through electronic price tags that can be updated much more quickly than paper price tags.

A potential disadvantage of dynamic prices is that they may be perceived as unfair by customers. The dual entitlement principle argues that firms are entitled to a fair profit and customers to a fair price~\citep{Kahnemann1986a,Kahnemann1986b}. %
To determine whether a particular price is fair customers may refer to price histories and prices charged by competitors~\citep{Briesch1997}, and research has shown that prices paid by other customers in particular have a major effect on customer satisfaction~\citep{Novemsky2004}. Customer satisfaction, on the other hand, is crucial for repeated purchase interactions~\citep{Barry1996}. For more detailed studies on the perceived fairness of yield management we refer to interested reader to the work of \citet{Kambil2001} and \citet{Mauri2007} and the references therein.

The potential negative effects of dynamic pricing are, at this point, very much a subject of academic debate and public controversy. The latter is witnessed for example by recent news items on spikes in surge prices offered by online ridesharing platform Uber on New Year's Eve of 2014 \citep{cbsnews} and on online retailer Amazon's increase of the price of bottled water in Florida amidst preparations for Hurricane Irma \citep{Forbes17}. 
These news items indicate that customers are becoming increasingly aware and wary of dynamic pricing, and that there is a tension between the potential for higher revenue achievable through dynamic prices and a decrease in customer satisfaction that may result from their use.
Our point of departure in this article is the natural question whether the advantages of dynamic pricing outweigh its disadvantages. We address this question by quantifying precisely the relative loss in revenue due to a use of static prices instead of dynamic ones.

\subsection{Our Contribution}

We provide the first formal study of the revenue gap between dynamic and static pricing in a multiple-goods setting. We focus on the case where $k$ homogeneous goods are sold sequentially to $n$ unit-demand buyers arriving one at a time. Motivated by situations in which the $n$ buyers that show up are random draws from a very large population, we focus on the case where buyer valuations are independently and identically distributed. We note that both the assumption of identical goods and the assumption of a large population of customers are reasonable for the industries that pioneered dynamic pricing such as the airline, hotel, and railway industries.

More formally, we assume 
that we have~$k$ homogeneous goods for sale and~$n$ buyers arrive one by one, each of them interested in purchasing a single good and with a valuation drawn independently of those of the other buyers from a known distribution~$F$. 
When buyer~$i$ arrives with valuation~$v_i$ we offer it a price~$p_{ij}$ that can depend on the number~$j$ of goods still available, on the distribution~$F$, and on~$i$, but not on~$v_i$. If $v_i\geq p_{ij}$ buyer~$i$ purchases the good at price~$p_{ij}$, otherwise it leaves without a purchase. The process continues as long as additional buyers and goods remain, and we obtain as overall revenue the sum of the prices at which goods have been sold.

Prices $p_{ij}$ for $1\leq i\leq n$ and $1\leq j\leq k$ maximizing the expected revenue can be determined by solving a dynamic program, but they can be very complicated and can in general not be given in closed form. Instead of \emph{dynamic prices} as above it may therefore make sense to offer the same \emph{static price}~$p$ to each of the buyers. Static prices are appealing due to their simplicity and fairness but may lead to a loss in revenue. Our goal here will be to investigate the extent of this possible loss. Denoting by $\Rdnk$ and $\Rsnk$ the maximum revenue from dynamic and static prices when there are~$n$ buyers and~$k$ goods we will specifically be interested in bounds on the revenue gap $\Rdnk/\Rsnk$. Obviously $\Rdnk\geq\Rsnk$ and thus $\Rdnk/\Rsnk\geq 1$.

\subsubsection{Results}

We derive tight bounds on the ratio between~$\Rdnk$ and~$\Rsnk$ both for general distributions and the important special case of regular distributions. Regular distributions were introduced by \citet{Myer81a} in the context of revenue-maximizing auctions. A distribution is regular if the revenue $R(q)$ obtained by selling to a single buyer with probability~$q$ is concave in~$q$.
The reader may consult \figref{fig:results_graph} for a graphical representation of our results and \figref{fig:results} for a comparison with bounds from prior work.

\begin{figure}
\centering
\begin{tikzpicture}
\begin{axis}[
axis x line=center,
axis y line=center,
xlabel style={right},
ylabel style={above},
xlabel={$k$},
ylabel={$\Rdnk/\Rsnk$},
axis line style = {-latex,thick},
xmax=110,
ymax=2.1]
\addplot[color=blue] coordinates {
( 1, 1.58198) ( 2, 1.37112) ( 3, 1.28873) ( 4, 1.2428 ) ( 5, 1.21281)
( 6, 1.19136) ( 7, 1.17509) ( 8, 1.16223) ( 9, 1.15175) (10, 1.143  ) 
(11, 1.13556) (12, 1.12914) (13, 1.12352) (14, 1.11855) (15, 1.11413)
(16, 1.11015) (17, 1.10654) (18, 1.10326) (19, 1.10026) (20, 1.0975 )
(21, 1.09494) (22, 1.09258) (23, 1.09038) (24, 1.08832) (25, 1.08639)
(26, 1.08459) (27, 1.08288) (28, 1.08128) (29, 1.07976) (30, 1.07832)
(31, 1.07696) (32, 1.07566) (33, 1.07443) (34, 1.07325) (35, 1.07213)
(36, 1.07105) (37, 1.07002) (38, 1.06903) (39, 1.06809) (40, 1.06718)
(41, 1.0663 ) (42, 1.06546) (43, 1.06465) (44, 1.06386) (45, 1.06311)
(46, 1.06238) (47, 1.06167) (48, 1.06099) (49, 1.06033) (50, 1.05969)
(51, 1.05907) (52, 1.05846) (53, 1.05788) (54, 1.05731) (55, 1.05676)
(56, 1.05622) (57, 1.0557 ) (58, 1.0552 ) (59, 1.0547 ) (60, 1.05422)
(61, 1.05375) (62, 1.05329) (63, 1.05285) (64, 1.05241) (65, 1.05199)
(66, 1.05157) (67, 1.05117) (68, 1.05077) (69, 1.05039) (70, 1.05001)
(71, 1.04964) (72, 1.04928) (73, 1.04892) (74, 1.04857) (75, 1.04823)
(76, 1.0479 ) (77, 1.04758) (78, 1.04726) (79, 1.04694) (80, 1.04663)
(81, 1.04633) (82, 1.04604) (83, 1.04575) (84, 1.04546) (85, 1.04518)
(86, 1.04491) (87, 1.04464) (88, 1.04437) (89, 1.04411) (90, 1.04386)
(91, 1.0436 ) (92, 1.04336) (93, 1.04311) (94, 1.04287) (95, 1.04264)
(96, 1.04241) (97, 1.04218) (98, 1.04195) (99, 1.04173)(100, 1.04152)
};
\addplot[color=green!50!black] coordinates {
( 1, 1.57275) ( 2, 1.36088) ( 3, 1.27745) ( 4, 1.23056) ( 5, 1.19966)
( 6, 1.17736) ( 7, 1.16028) ( 8, 1.14665) ( 9, 1.13542) (10, 1.12594)
(11, 1.1178 ) (12, 1.11068) (13, 1.10438) (14, 1.09875) (15, 1.09366)
(16, 1.08903) (17, 1.08478) (18, 1.08086) (19, 1.07721) (20, 1.07381)
(21, 1.07062) (22, 1.06761) (23, 1.06476) (24, 1.06206) (25, 1.05948)
(26, 1.05701) (27, 1.05464) (28, 1.05236) (29, 1.05016) (30, 1.04802)
(31, 1.04595) (32, 1.04393) (33, 1.04195) (34, 1.04001) (35, 1.0381 )
(36, 1.03622) (37, 1.03435) (38, 1.03249) (39, 1.03063) (40, 1.02877)
(41, 1.02688) (42, 1.02497) (43, 1.02301) (44, 1.02097) (45, 1.01884)
(46, 1.01656) (47, 1.01406) (48, 1.01117) (49, 1.00749) (50, 1      )
(51,1)(52,1)(53,1)(54,1)(55,1)(56,1)(57,1)(58,1)(59,1)(60,1)
(61,1)(62,1)(63,1)(64,1)(65,1)(66,1)(67,1)(68,1)(69,1)(70,1)
(71,1)(72,1)(73,1)(74,1)(75,1)(76,1)(77,1)(78,1)(79,1)(80,1)
(81,1)(82,1)(83,1)(84,1)(85,1)(86,1)(87,1)(88,1)(89,1)(90,1)
(91,1)(92,1)(93,1)(94,1)(95,1)(96,1)(97,1)(98,1)(99,1)(100,1)
};
\addplot[color=red] coordinates {
( 1, 1.55881) ( 2, 1.34528) ( 3, 1.26009) ( 4, 1.21147) ( 5, 1.1789 )
( 6, 1.15493) ( 7, 1.13618) ( 8, 1.12085) ( 9, 1.10789) (10, 1.09661)
(11, 1.08657) (12, 1.07745) (13, 1.06899) (14, 1.061  ) (15, 1.05328)
(16, 1.04563) (17, 1.0378 ) (18, 1.02936) (19, 1.01923) (20, 1)
(21,1)(22,1)(23,1)(24,1)(25,1)(26,1)(27,1)(28,1)(29,1)(30,1)
(31,1)(32,1)(33,1)(34,1)(35,1)(36,1)(37,1)(38,1)(39,1)(40,1)
(41,1)(42,1)(43,1)(44,1)(45,1)(46,1)(47,1)(48,1)(49,1)(50,1)
(51,1)(52,1)(53,1)(54,1)(55,1)(56,1)(57,1)(58,1)(59,1)(60,1)
(61,1)(62,1)(63,1)(64,1)(65,1)(66,1)(67,1)(68,1)(69,1)(70,1)
(71,1)(72,1)(73,1)(74,1)(75,1)(76,1)(77,1)(78,1)(79,1)(80,1)
(81,1)(82,1)(83,1)(84,1)(85,1)(86,1)(87,1)(88,1)(89,1)(90,1)
(91,1)(92,1)(93,1)(94,1)(95,1)(96,1)(97,1)(98,1)(99,1)(100,1)
};
\addplot[color=blue,dashed] coordinates {
( 1,2)( 2,2)( 3,2)( 4,2)( 5,2)( 6,2)( 7,2)( 8,2)( 9,2)(10,2)
(11,2)(12,2)(13,2)(14,2)(15,2)(16,2)(17,2)(18,2)(19,2)(20,2)
(21,2)(22,2)(23,2)(24,2)(25,2)(26,2)(27,2)(28,2)(29,2)(30,2)
(31,2)(32,2)(33,2)(34,2)(35,2)(36,2)(37,2)(38,2)(39,2)(40,2)
(41,2)(42,2)(43,2)(44,2)(45,2)(46,2)(47,2)(48,2)(49,2)(50,2)
(51,2)(52,2)(53,2)(54,2)(55,2)(56,2)(57,2)(58,2)(59,2)(60,2)
(61,2)(62,2)(63,2)(64,2)(65,2)(66,2)(67,2)(68,2)(69,2)(70,2)
(71,2)(72,2)(73,2)(74,2)(75,2)(76,2)(77,2)(78,2)(79,2)(80,2)
(81,2)(82,2)(83,2)(84,2)(85,2)(86,2)(87,2)(88,2)(89,2)(90,2)
(91,2)(92,2)(93,2)(94,2)(95,2)(96,2)(97,2)(98,2)(99,2)(100,2)
};
\addplot[color=green!50!black,dashed] coordinates {
( 1,1.98)( 2,1.96)( 3,1.94)( 4,1.92)( 5,1.90)
( 6,1.88)( 7,1.86)( 8,1.84)( 9,1.82)(10,1.80)
(11,1.78)(12,1.76)(13,1.74)(14,1.72)(15,1.70)
(16,1.68)(17,1.66)(18,1.64)(19,1.62)(20,1.60)
(21,1.58)(22,1.56)(23,1.54)(24,1.52)(25,1.50)
(26,1.48)(27,1.46)(28,1.44)(29,1.42)(30,1.40)
(31,1.38)(32,1.36)(33,1.34)(34,1.32)(35,1.30)
(36,1.28)(37,1.26)(38,1.24)(39,1.22)(40,1.20)
(41,1.18)(42,1.16)(43,1.14)(44,1.12)(45,1.10)
(46,1.08)(47,1.06)(48,1.04)(49,1.02)(50,1.00)
(51,1)(52,1)(53,1)(54,1)(55,1)
(56,1)(57,1)(58,1)(59,1)(60,1)
(61,1)(62,1)(63,1)(64,1)(65,1)
(66,1)(67,1)(68,1)(69,1)(70,1)
(71,1)(72,1)(73,1)(74,1)(75,1)
(76,1)(77,1)(78,1)(79,1)(80,1)
(81,1)(82,1)(83,1)(84,1)(85,1)
(86,1)(87,1)(88,1)(89,1)(90,1)
(91,1)(92,1)(93,1)(94,1)(95,1)
(96,1)(97,1)(98,1)(99,1)(100,1)
};
\addplot[color=red,dashed] coordinates {
( 1,1.95)( 2,1.90)( 3,1.85)( 4,1.80)( 5,1.75)
( 6,1.70)( 7,1.65)( 8,1.60)( 9,1.55)(10,1.50)
(11,1.45)(12,1.40)(13,1.35)(14,1.30)(15,1.25)
(16,1.20)(17,1.15)(18,1.10)(19,1.05)(20,1.00)
(21,1)(22,1)(23,1)(24,1)(25,1)
(26,1)(27,1)(28,1)(29,1)(30,1)
(31,1)(32,1)(33,1)(34,1)(35,1)
(36,1)(37,1)(38,1)(39,1)(40,1)
(41,1)(42,1)(43,1)(44,1)(45,1)
(46,1)(47,1)(48,1)(49,1)(50,1)
(51,1)(52,1)(53,1)(54,1)(55,1)
(56,1)(57,1)(58,1)(59,1)(60,1)
(61,1)(62,1)(63,1)(64,1)(65,1)
(66,1)(67,1)(68,1)(69,1)(70,1)
(71,1)(72,1)(73,1)(74,1)(75,1)
(76,1)(77,1)(78,1)(79,1)(80,1)
(81,1)(82,1)(83,1)(84,1)(85,1)
(86,1)(87,1)(88,1)(89,1)(90,1)
(91,1)(92,1)(93,1)(94,1)(95,1)
(96,1)(97,1)(98,1)(99,1)(100,1)
};
\addplot[color=black,thick] coordinates {
(1,1) (100,1)
};
\end{axis}
\end{tikzpicture}	
\caption{Revenue gaps $\Rdnk / \Rsnk$ for $n=20$ buyers (red), $n=50$ buyers (green), and $n\to \infty$ buyers (blue). The results for regular distributions are drawn with solid lines; the results for arbitrary distributions are dashed. For better illustration, the lines are smooth even though $k$ only takes integral values.\label{fig:results_graph}}
\end{figure}
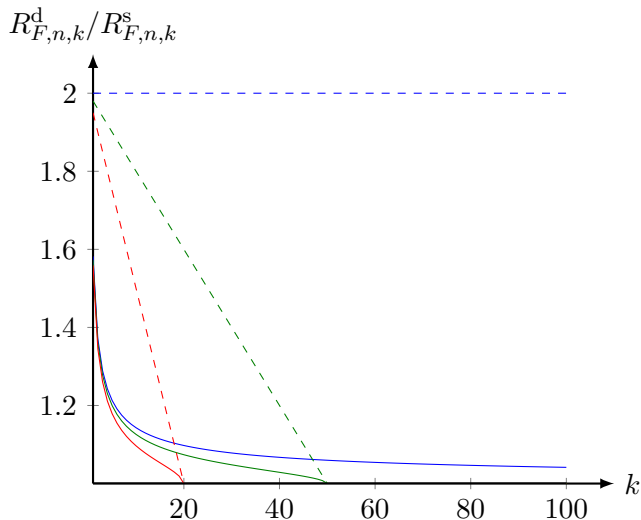

\begin{figure}[t]
\centering
\begin{tabular}{@{}p{2.1cm}p{.2cm}p{.6cm}p{0cm}p{.6cm}p{.5cm}p{3.5cm}p{2.5cm}}
\toprule
&& \multicolumn{1}{c}{$k$} && \multicolumn{1}{c}{$n$} && \multicolumn{2}{c}{~~~~distribution~$F$}\\
&& && && \multicolumn{1}{c}{\textbf{regular}} & \multicolumn{1}{c}{\textbf{irregular}} \\
\midrule
\multirow{2}{*}{\textbf{this article}} && \multicolumn{1}{c}{$\geq 1$} && \multicolumn{1}{c}{$\geq k$} & & \multicolumn{1}{c}{$\frac{k}{\E_{X\sim\textsf{Bin}(n,\frac{k}{n})}[\min\{X,k\}]}$} &  \multicolumn{1}{c}{$2-\frac{k}{n}$} \\[12pt]
&& \multicolumn{1}{c}{$\geq 1$} && \multicolumn{1}{c}{$\to\infty$} & & \multicolumn{1}{c}{$\frac{1}{1-\frac{k^k}{e^{k}k!}}$} &  \multicolumn{1}{c}{$2$}\\
\midrule
\multirow{1}{*}{\textbf{prior work}} && \multicolumn{1}{c}{$=1$} && \multicolumn{1}{c}{$\to\infty$} & & \multicolumn{1}{c}{$\bigl[1.11,\frac{e}{e-1}\bigr]$} & \multicolumn{1}{c}{$2$} \\
\bottomrule
\end{tabular}
\caption{Comparison between bounds on the revenue gap $\Rdnk / \Rsnk$ of dynamic and static prices obtained in this article and the best bounds from prior work. Numbers in square brackets indicate lower and upper bounds, a single number means that the bound is tight.}
\label{fig:results}
\end{figure}

For the special case of regular distributions we obtain an upper bound on the revenue gap of at most $k/\E_{X\sim\textsf{Bin}(n,k/n)}[\min\{X,k\}]$, where $\textsf{Bin}(n,k/n)$ denotes a binomial distribution with~$n$ trials and success probability~$k/n$. This bound increases in~$n$ and tends to $1/(1- \frac{k^k}{e^{k}k!}) \approx (1/(1- \frac{1}{\sqrt{2 \pi k}}))$. %
In the single-good case our upper bound simplifies to $1/(1-(1-\frac{1}{n})^n)$ and asymptotically to $e/(e-1)\approx 1.582$, which also follows from an earlier result of \citet{ChawlaHMS10}. 
Our upper bound is in fact tight for all values of~$n$ and $k$, which we show by optimizing over a specific family of distributions that are continuous on an interval with additional point mass on the highest value.
A tight result was not previously known even in the single-good case~\citep{BlumrosenH08,Hartline17}.

Our upper bound of $k/\E_{X\sim\textsf{Bin}(n,k/n)}[\min\{X,k\}]$ arises from a balls-and-bins analogy for the setting with a static price. There are two bins, respectively counting buyers who are willing and unwilling to buy. The price is set in such a way that each buyer is willing to buy with probability $k/n$, so that the expected number of buyers willing to buy is equal to $\E_{X \sim \textsf{Bin}(n,k/n)}[X]$ and the expected number of items sold to $\E_{X\sim\textsf{Bin}(n,k/n)}[\min\{X,k\}]$. The asymptotic bound of $1/(1- \frac{k^k}{e^{k}k!})\approx (1/(1- \frac{1}{\sqrt{2 \pi k}}))$ then follows from a standard approximation of the binomial distribution by the Poisson distribution.
Similar asymptotic bounds for related but different problems were obtained by \citet{Alaei14} and \citet{Chakraborty2010}.\footnote{\citet{Alaei14} studies posted-price mechanisms for revenue maximization in a setting with~$k$ homogeneous goods and~$n$ unit-demand buyers whose values are independent but not necessarily identically distributed. The mechanisms considered by \citeauthor{Alaei14} allow individual buyers to be skipped, which cannot be implemented with a single static price. For the same setting, \citet{Chakraborty2010} investigate the revenue gap between pricing policies that may change the order in which buyers are seen after some of the buyers have arrived and pricing policies that use a fixed ordering. Since even their fixed-ordering policies may offer different prices to different buyers, their results are incomparable to ours.}

For general distributions we obtain an upper bound on the revenue gap of $2-k/n$. This bound is again tight for all values of~$n$ and~$k$, which we show by means of a discrete distribution with support of size two. To our knowledge these are the first bounds on the revenue gap between static and dynamic prices for irregular distributions and more than one good. For a single good an asymptotic lower bound of $2$ was given by \citet{ChawlaHK07}, and an upper bound of~$2$ can be obtained via prophet inequalities~\citep[\eg][]{Hartline17}.

The main take-away from our analysis is the qualitatively very different behavior of the revenue gap as $n$ and $k$ grow large. In many cases of interest $\Rdnk/\Rsnk \rightarrow 1$ for regular distributions and so static pricing is ``first best'', while $\Rdnk/\Rsnk \rightarrow 2$ for irregular distributions and so there is a persistent gap between dynamic and static pricing.
Many common distributions are regular \citep[cf.][]{Ewerhart2013}, whereas regularity may fail even for very simple mixtures of regular distributions. For instance the uniform distribution on $[0,1]$ is regular while the distribution that draws with probability $4/5$ from a uniform distribution on $[0,4]$ and with probability $1/5$ from a uniform distribution on $[4,10]$ is not, see \figref{fig:uniform}. This leads to the following implication for the large population model and large $n$ and $k$: \emph{If the population is homogeneous, static pricing is optimal. Otherwise, dynamic pricing may be significantly better.}

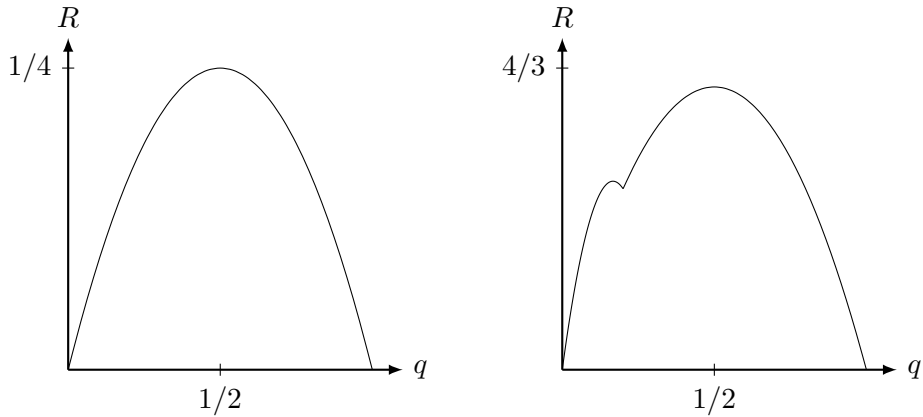
\begin{figure}[tb]
\centering
\begin{subfigure}{0.4\textwidth}
\begin{tikzpicture}[scale=4]
\draw[thick,-latex] (0,0) -- (1.1,0) node[right] {$q$};
\draw[thick,-latex] (0,0) -- (0,1.1) node[above] {$R$};
\draw (0.5,0.02) -- (0.5,-0.02) node[below] {$1/2$};
\draw (0.02,1) -- (-0.02,1) node[left] {$1/4$};
\draw[domain=0:1,smooth,variable=\q]  plot ({\q},{4*\q*(1-\q)});
\end{tikzpicture}
\end{subfigure}
\begin{subfigure}{0.4\textwidth}
\begin{tikzpicture}[scale=4]
\draw[thick,-latex] (0,0) -- (1.1,0) node[right] {$q$};
\draw[thick,-latex] (0,0) -- (0,1.1) node[above] {$R$};
\draw (0.5,0.02) -- (0.5,-0.02) node[below] {$1/2$};
\draw (0.02,1) -- (-0.02,1) node[left] {$4/3$};
\draw[domain=0:0.2,smooth,variable=\q]  plot ({\q},{0.75*\q*(10-30*\q)});
\draw[domain=0.2:1,smooth,variable=\q]  plot ({\q},{0.75*\q*(5-5*\q)});
\end{tikzpicture}
\end{subfigure}

\caption{Revenue curves for the uniform distribution on $[0,1]$ (left), and the distribution that draws with probability $4/5$ uniformly from $[0,4]$ and with the remaining probability uniformly from $[4,10$] (right).\label{fig:uniform}}
\end{figure}

\subsubsection{Techniques}

The main difficulty in analyzing the revenue gap is the unavailability of closed-form expressions for the optimal dynamic prices and the corresponding revenue, which holds already for very simple cases like valuations distributed uniformly on $[0,1]$ \citep[\S~5a]{GilbertM66}. We address this difficulty by considering an ex-ante relaxation of the dynamic program characterizing the optimal dynamic prices which relaxes the constraint that at most~$k$ goods are sold to only hold in expectation, and by appealing to and exploiting properties of Bernstein polynomials in the analysis of the revenue gap between static prices and this relaxation.

We begin our analysis by observing that for independent and identically distributed valuations, the ex-ante relaxation has a symmetric optimum where each buyer obtains a good with the same probability $q^*$. In the regular case it then holds that $R(q^*)=q^*p^*$ for some price $p^*$, and we analyze the revenue gap between the ex-ante relaxation and the static price~$p^*$. Our main idea in this context is to write the revenue generated by the static price~$p^*$ as the Bernstein polynomial of a monotone and concave function. By using that the approximation by Bernstein polynomials preserves monotonicity and concavity, we are then able to establish that the revenue gap is maximized when $q^*=k/n$. A simplified version of the resulting upper bound can be obtained from the well-known fact that the binomial distribution with~$n$ trials and success probability~$k/n$ converges to the Poisson distribution with rate~$k$.
To obtain a matching lower bound we consider a family of distributions with a piecewise linear revenue curve that is monotone and, for some parameter $r$, bends at a revenue of $r$. We write the revenue gap for distributions within this family as a maximin problem where in the upper level we choose~$r$ and in the lower level an adversary chooses the probability~$q$ with which a single buyer accepts the optimal static price. We then use Berge's maximum theorem to establish the existence of a choice of~$r$ for which the minimum in the lower level is attained at $q=k/n$, and observe that the resulting lower bound is equal to the upper bound for all values of~$n$ and~$k$.

In the case of general distributions a symmetric solution of the ex-ante relaxation may correspond to a lottery over prices, and the methods described so far would only allow us to bound the revenue gap between the ex-ante relaxation and a static \emph{lottery over two prices}. To bound the revenue gap for a \emph{single} static price we instead consider a different lottery that selects one of the prices for all buyers, and carefully choose the probabilities within the lottery to express its expected revenue in terms of a convex combination of Bernstein polynomials of a concave function. By using further properties of Bernstein polynomials we are then able to determine the worst case for the upper bound, which occurs when the ex-ante relaxation sells to $n-k$ buyers with a vanishing probability and to~$k$ buyers with probability close to~$1$. A simple two-point distribution shows that the resulting bound is again tight for all values of $k$ and $n$. 

The ex-ante relaxation we consider, like previous relaxations of this type, in fact turns out to relax not only the optimal dynamic prices but any incentive-compatible mechanism.\footnote{The ex-ante relaxation was identified as a quantity of interest by \citet{ChawlaHK07} and used subsequently for example by \citet{Yan11}, \citet{Alaei14}, and \citet{ChMi16a}.} Consequently our upper bounds apply even to the gap between the revenue from the best incentive-compatible mechanism and the best static price. Our lower bounds show that the upper bounds are tight already for the gap between dynamic and static prices. 

\subsection{Related Work}

Dynamic pricing is one of the classic topics in operations research. Work in this area has emphasized more complex models and has focused on characterizing and computing optimal prices in these models \citep{GallegovanRyzin94,GallegovanRyzin97,BitranMondschein97,SubramanianEtAl99,FengXiao00,AvivPazgal05,BesbesZeevi09,BesbesLobel15,Caldentey17}.
This literature, however, has not quantified how much revenue can be gained by using dynamic rather than static prices.

A question related to ours was first considered by \citet{ChawlaHK07}, who used virtual values to compare the maximum revenue from an incentive-compatible mechanism to that of sequential posted prices. They gave dynamic prices approximating the optimal revenue from a mechanism within a factor of~$3$. For identical regular distributions the prices are in fact identical and an improved upper bound of~$2.17$ can be obtained. For identical irregular distributions the authors gave a lower bound of~$2$. It is not difficult to see that both the upper bound of~$2.17$ and the lower bound of~$2$ apply to the gap between dynamic and static pricing we consider here.

For the same problem and under the same assumptions, \citet{ChawlaHMS10} obtained tighter bounds by analyzing virtual values using results from optimal stopping theory, in particular the classic prophet inequality of \citet{KrengelS1973,KrengelS1978} 
and an improved version due to the authors for a setting where buyers can be considered in a particular order. When specialized to identical regular distributions, these results imply an upper bound of $e/(e-1)$ on the gap between dynamic and static prices. For general distributions the technique requires an ironing step, and thus additional work to obtain an upper bound with respect to a single static price rather than a static lottery over prices.

Non-identical distributions were first analyzed by \citet{AlaeiHNPY15}, who formulated the revenue gap between the optimal incentive-compatible mechanism and a static price as a mathematical program and through a series of relaxations obtained an upper bound of $e\approx 2.718$ for the case of regular distributions. 
In parallel work to our work a tight bound of roughly~$2.62$ was shown on the gap between dynamic and static prices for regular distributions~\citep{Jin18}, while for general distributions this gap can be as large as~$n$~\citep{AlaeiHNPY15}. 

\citet{BlumrosenH08} gave a characterization of the optimal dynamic and static prices for a single good and identical distributions, along with formulas for the asymptotic revenue in terms of parameters of the distribution. For the regular power-law distribution they showed that the maximum revenue obtained respectively by an incentive-compatible mechanism, by dynamic prices, and by a static price is $0.89\sqrt{n}$, $0.71 \sqrt{n}$, and $0.64\sqrt{n}$, which implies an asymptotic lower bound of approximately~$0.71/0.64 \approx 1.11$ for the gap we are interested in. The solution of an exercise set by \citet{Hartline17} provides an improved, yet not a tight, lower bound for the single-good case.

All of these works on the gap between static and dynamic pricing belong to the rapidly expanding literature on prophet inequalities   \citep{HajiaghayiKS07,ChawlaHMS10,KleinbergW12,AzarKW14,Alaei14,DuttingK15,FeldmanGL15,Rubi16,RuSi17,DuettingFKL17,CorreaFHOV17,AbolhassaniEEHK17,EhsaniHKS18,CorreaEtAl19}. 

\section{Preliminaries}

We consider a situation with~$k$ homogeneous goods and a set $[n]=\{1,\dots,n\}$ of unit-demand buyers, where $k\geq n$. For $i\in[n]$, we denote by~$v_i$ the valuation of buyer~$i$ for receiving a good, and assume that valuations are distributed independently according to the same non-negative, and possibly discontinuous, probability distribution with cumulative distribution function~$F$. 

Denote by $\F(x)$ the left-sided limit of~$F$ at~$x$, \ie $\F(x)=\lim_{y\uparrow x}F(y)$. The \emph{revenue curve} of~$F$ then is the function $R:[0,1]\to\R_{\geq 0}$ such that $R(q)=q\F^{-1}(1-q)$, where, with a slight abuse of notation, $\F^{-1}(1-q)=\max\{p\geq 0\midd\F(p)\leq 1-q\}$. By construction $\F$ is left-continuous, so that the maximum is attained.  Intuitively, $R(q)$ is the maximum revenue that can be obtained by selling a single good to a single buyer at a deterministic price so that the probability of sale is exactly~$q$.\footnote{Note that when~$F$ is not continuous, a price that leads to a probability of sale of exactly~$q$ may not exist. In this case $R(q)$ is the revenue obtained by selling at the largest price that sells with probability at least $q$, but only with probability $q$.}
We call distribution~$F$ \emph{regular} if~$R$ is concave. Note that for distributions with a non-vanishing density $f(x) = \F'(x)$ regularity is equivalent to the fact that
\begin{align*}
R'(q) = \bigl( q \F^{-1}(1-q) \bigr)' = \F^{-1}(1-q) - \frac{q}{f(\F^{-1}(1-q))}
\end{align*}
is non-decreasing for all $q \in [0,1]$.
Substituting $x = \F^{-1}(1-q)$ we obtain that 
\begin{align*}
\phi(x) = x - \frac{1-\F(x)}{f(x)}
\end{align*}
is non-decreasing for all $x$. Our notion of regularity is thus equivalent to that of \citep{Myer81a}, who calls~$\phi$ the \emph{virtual value function} of~$F$, but is defined also for distributions without a density.

For irregular distributions it will be useful to also consider an \emph{ironed revenue curve} $\tR:[0,1]\to\R_{\geq 0}$ given by 
\begin{align*}
\tR(q)=\sup_{a\in[0,q],b\in[q,1]} \frac{(b-q)R(a)+(q-a)R(b)}{b-a}.
\end{align*}
Intuitively, $\tR(q)$ is the maximum revenue that can be obtained by selling a single good to a single buyer using a lottery over prices so that the probability of sale is exactly~$q$. 

A \emph{sequential posted-price mechanism} considers each of the buyers in turn, and as long as a good is available offers the current buyer a price at which a good can be purchased. The buyer purchases a good if its valuation exceeds the price it is being offered, otherwise it leaves without a good. The price offered to buyer~$i$ may depend on the realized values~$v_j$ of buyers $j<i$ and on distribution~$F$, but not on the realized value~$v_i$ of buyer~$i$. The revenue of a mechanism then is the expected sum of the prices at which the goods are sold, where the expectation is taken over the joint distribution of buyers' valuations. 

We distinguish between dynamic prices, which may depend both on~$i$ and the number~$j$ of goods still available when buyer~$i$ arrives, and static prices, which may depend on neither of the two. \emph{Dynamic prices} can be written as an $n\times k$ matrix $\vec p=(p_{i,j})_{i\in[n],j\in[k]}$, where $p_{i,j}$ is the price offered to buyer~$i$ when~$j$ goods are still available. Denoting by $\Q_{i,j}(\vec p)$ the probability that exactly~$j$ goods are available when buyer~$i$ is considered, the revenue obtained by $\vec p$ is \begin{align*}
	\Rdnk(\vec p) = \sum_{i\in[n]} \sum_{j \in [k]} p_{i,j} \bigl(1- \F(p_{i,j})\bigr) \Q_{i,j}(\vec p).
\end{align*}
A \emph{static price}~$p\in\R$ obtains a revenue of
\begin{align*}
	\Rsnk(p) = \sum_{j = 1}^n \min\{j,k\} \,p\, \binom{n}{j} \F(p)^{n-j} (1-\F(p))^{j} .
\end{align*}
We denote by $\Rdnk$ and $\Rsnk$ the maximum revenue from static and dynamic prices, \ie 
\begin{align*}
	\Rdnk &= \max_{\mathclap{\vec p\in\R_{\geq 0}^{n \times k}}}\;\Rd_n(\vec p) & \text{and} && \Rsnk &= \max_{\mathclap{p\in\R_{\geq 0}}}\,\Rs_n(p) .
\end{align*}

Both static and dynamic prices are incentive compatible in the sense that they make it optimal for each buyer to implicitly reveal its true valuation, by purchasing a good if and only if the buyer's valuation exceeds the price being offered. We may more generally consider mechanisms which solicit a bid from each buyer and based on the bids determine an allocation of goods to buyers and prices the buyers have to pay. An \emph{incentive-compatible mechanism} then is a mechanism in which it is optimal for each buyer to bid its true valuation assuming that the other buyers do the same, in the sense that this maximizes the expected difference between the buyer's value from being assigned a good and its payment, where the expectation is taken over the distribution of the values of the other buyers.\footnote{Following the seminal work of \citet{Myer81a} on revenue-maximizing incentive-compatible mechanisms for the case of a single good, such mechanisms are also referred to as optimal auctions. We use the standard notion of incentive compatibility where truthful revelation of valuations is a Bayes-Nash equilibrium.} We denote by $\Rm_{n,k}$ the maximum revenue from an incentive-compatible mechanism and note that for all $F$, $n$ and $k$,
\[
	\Rsnk \leq \Rdnk \leq \Rmnk .
\]

For ease of exposition we will assume all three quantities to exist and be finite. All results can be seen to hold in general by standard limit arguments. 

\section{An Ex-Ante Relaxation}
\label{sec:ex-ante}

As the revenue from dynamic prices is rather difficult to analyze directly, it will be useful to consider instead an \emph{ex-ante relaxation} of the associated maximization problem with revenue
\begin{align*}
	\Rxnk = \max\Biggl\{\sum_{i\in[n]} \tR(q_i) \midd \vec q\in[0,1]^n, \sum_{i\in[n]} q_i \leq k \Biggr\}.
\end{align*}
Intuitively this quantity relaxes the requirement that at most~$k$ goods can be sold to only hold in expectation over the random draws from~$F$. 
It is important to note that the definition of~$\Rxnk$ uses the ironed revenue curve~$\tR$. In the case of regular distributions~$R$ and~$\tR$ are identical, so that a revenue of $\Rxnk$ could indeed be obtained by offering a good independently to each buyer at a deterministic price and selling at most~$k$ goods in expectation. In the case of irregular distributions a lottery over prices may be required to achieve the same.

In the following example we see that all four quantities $\Rsnk$, $\Rdnk$, $\Rmnk$, and $\Rxnk$ are distinct and the ex-ante relaxation provides an upper bound not only on the maximum revenue from dynamic prices but also on that from any incentive-compatible mechanism.

\begin{example}[A single good and two buyers with uniform valuations]
\label{ex:uniform}
	Consider a situation with one good and two buyers with valuations drawn independently and uniformly from $[0,1]$, i.e., $F(x) = x$ for $\max\{0,x\}$ for $x \leq 1$ and $F(x) = 1$ for $x > 1$. It is not difficult to see that
	\begin{align*}
		\Rs_{F,2,1} &= \max\Bigl\{p \bigl( 1-p^2\bigr) \midd p \in \R_{\geq 0}\Bigr\} = \frac{2}{3\sqrt{3}} \approx 0.385,\\
		\Rd_{F,2,1} &= \max\Bigl\{ p_{1,1} (1-p_{1,1}) + p_{2,1} (1-p_{2,1})p_{1,1}\midd \vec{p} \in \R_{\geq 0}^{2 \times 1}\Bigr\} = \frac{25}{64} \approx 0.391.
	\end{align*}
	The revenue-optimal incentive-compatible mechanism is a second-price auction with a reserve price $p^*$ that corresponds to a probability maximizing the revenue curve, i.e., $p^* = 1/2$, see also Figure~\ref{fig:uniform} for an illustration \citep{Myer81a}. To compute the revenue of the optimal mechanism let $X_1,X_2$ be drawn from the uniform distribution on $[0,1]$, and let $Y_1=\min\{X_1,X_2\}$ and $Y_2=\max\{X_1,X_2\}$. Then
	\begin{align*}
		\Rm_{F,2,1} &= \frac{1}{2} \Pr\biggl[Y_1 \leq \frac{1}{2}, Y_2 \geq \frac{1}{2}\biggr] + \Pr\biggl[Y_1 \geq \frac{1}{2}\biggr]\E\biggl[Y_1 \;\bigg|\; Y_1 \geq \frac{1}{2}\biggr] = \frac{1}{2}\cdot \frac{1}{2} + \frac{1}{4}\cdot\frac{2}{3} = \frac{5}{12} \approx 0.417.
	\end{align*} 
	For the ex-ante relaxation, we obtain
	\begin{align*}
	\Rx_{F,2,1} = \max \Bigl\{ q_1(1-q_1) + q_2(1-q_2) : \vec q \in [0,1]^2, q_1 + q_2 \leq 1 \Bigr\} = \frac{1}{2} = 0.5.	
	\end{align*}
	We see that indeed in this example $\Rs_{F,2,1} < \Rd_{F,2,1} < \Rm_{F,2,1} < \Rx_{F,2,1}$ as claimed. $\quad \triangle$
\end{example}

The following lemma shows that the ex-ante relaxation is always an upper bound on the optimal auction revenue and, thus, also on the optimal revenue of dynamic pricing. Its uses an argument analogous to those in the proofs of Theorem~1 of \citet{Alaei14} and Lemma~3.1 of \citet{ChMi16a}.
\begin{restatable}{lemma}{lemRelaxation}
\label{lem:relaxation}
Let $n,k\in\N$ with $k\leq n$ and $F$ be arbitrary. Then there exists $\vec q\in[0,1]^n$ such that $\sum_{i=1}^n q_i\leq k$ and $\Rmnk \leq \sum_{i=1}^n \tR(q_i)$.
\end{restatable}

\begin{proof}
	By definition of $\Rmnk$ there exists an incentive-compatible mechanism~$M$ that sells at most~$k$ goods to~$n$ buyers and achieves revenue $\Rm_{n,k}$. Formally $M=(\vec x,\vec p)$ for two functions $\vec x:\R_{\geq 0}^n\to[0,1]^n$ and $\vec p:\R_{\geq 0}^n\to\R_{\geq 0}^n$, such that for any vector $\vec v\in\R_{\geq 0}$ of valuations, $(\vec x(\vec v))_i$ is the probability that~$M$ assigns a good to buyer~$i$ and $(\vec p(\vec v))_i$ is the payment of buyer~$i$. Denote by~$q_i=\E_{\vec v\sim F^n}(\vec x(\vec v))_i$ the probability that mechanism~$M$ sells a good to buyer~$i$ and observe that $\sum_{i=1}^n q_i\leq k$.
	
	For each buyer~$i$ we can construct from~$M$ a mechanism~$M_i$ that sells at most one good to the single buyer~$i$: given report~$v_i$ from buyer~$i$ the mechanism draws $(v_{j})_{j\in[n]\setminus\{i\}}$ from~$F^{n-1}$, assigns a good to buyer~$i$ with probability $x_i(\vec v)$ and charges it a payment of $p_i(\vec v)$.
	
	Denote by~$R_i$ the expected revenue achieved by~$M_i$. Since the expected revenue from buyer~$i$ is the same under $M$ and $M_i$, $\Rmnk=\sum_{i=1}^n R_i$. Mechanisms~$M$ and~$M_i$ are indistinguishable from the point of view of buyer~$i$, so $M_i$ is incentive compatible. This implies that $R_i\leq \tR(q_i)$, because~$M_i$ sells to buyer~$i$ with probability~$q_i$ and because $\tR(q_i)$ is the maximum revenue of any incentive-compatible mechanism that sells to a single buyer with probability~$q_i$~\citep[\eg][]{Hartline17}. In summary
\[
	\Rmnk = \sum_{i=1}^n R_i \leq \sum_{i=1}^n \tR(q_i),
\]
as claimed.
\end{proof}

A useful property of the ex-ante relaxation in the case of identical distributions is the existence of a symmetric optimum, as made precise by the following lemma. %
\begin{restatable}{lemma}{lemSymmetric}  \label{lem:symmetric}  %
	Let $n,k\in\N$ with $k\leq n$ and $F$ be arbitrary. Then there exists $q^*\in[0,k/n]$ such that $\Rxnk=n\tR(q^*)$.
\end{restatable}

\begin{proof}
	The objective of the ex-ante relaxation is a sum of concave functions and therefore concave. Both the objective and the constraints are invariant under permutations of the set of buyers. Consider an optimal solution~$\vec q$ of the ex-ante relaxation, and let $\bar{\vec q} = \frac{1}{n!} \sum_{\pi\in S_n} (P_{\pi} \vec q)$, where~$S_n$ is the set of permutations of~$[n]$ and~$P_\pi$ is the permutation matrix corresponding to permutation~$\pi$. Clearly $\bar{\vec q}$ is feasible. Moreover $\bar{\vec q}$ is invariant under permutations, so there must exist some $\alpha\in\R_{\geq 0}$ such that $\bar{\vec q}=\alpha\vec{1}$. In fact, as~$\bar{\vec q}$ is feasible, $\alpha\leq k/n$. Finally, by concavity of the objective and by Jensen's inequality,
\begingroup 
	\begin{align*}
	\sum_{i \in [n]} R(\bar{q}_i) &= \sum_{i \in[n]} R\biggl(\frac{1}{n!} \sum_{\pi \in S_n} (P_\pi \vec q)_i\biggr) 
	\geq \sum_{i \in [n]} \frac{1}{n!} \sum_{\pi \in S_n} R\bigl((P_\pi \vec q)_i\bigr) \\
	&= \sum_{\mathclap{\pi \in S_n}} \, \frac{1}{n!} \sum_{i \in [n]} R\bigl((P_\pi \vec q)_i\bigr) 
	= \sum_{\mathclap{\pi \in S_n}} \, \frac{1}{n!} \sum_{i \in [n]} R(q_i) = \sum_{\mathclap{i \in [n]}} R(q_i),
	\end{align*} 
	so $\bar{\vec q}$ is optimal.
\endgroup
\end{proof}

\section{Regular Distributions}
\label{sec:regular}

Our first objective will be to bound the revenue gap between dynamic and static prices for the special case of regular distributions. We obtain the following upper bound, which will turn out to be tight for all values of~$n$ and~$k$.  \par

\begin{theorem}  \label{thm:regular_upper}
	For any number~$n$ of buyers with values drawn independently from a regular distribution~$F$ and any number $k\leq n$ of goods,
\[
	\frac{\Rxnk}{\Rsnk} 
	\leq \frac{k}{\E_{X \sim \mathsf{Bin}(n,k/n)}[\min\{X,k\}]} ,
\]
where $\mathsf{Bin}(n,k/n)$ denotes a binomial distribution with $n$ trials and success probability $k/n$.
\end{theorem}
Before proving this result we also give a simplified bound that can be obtained from well-known properties of the approximation of a Poisson distribution by a series of binomial distributions. The development is standard and summarized in \appref{app:poisson} for completeness.
\begin{restatable}{corollary}{corPoisson}  \label{cor:poisson}
	For any number~$n$ of buyers with values drawn independently from a regular distribution~$F$ and any number $k\leq n$ of goods,
	\[
		\frac{\Rxnk}{\Rsnk} \leq \frac{1}{1 - \frac{k^k}{e^{k}k!}} .
	\]
\end{restatable}
\noindent When~$k$ is large, then $k!\approx\sqrt{2\pi k}\bigl(\frac{k}{e}\bigr)^{\!k}$ by Stirling's formula, so $\Rxnk/\Rsnk \approx 1/(1- \frac{1}{\sqrt{2 \pi k}})$. 

\subsection{Prerequisites for the Proof of \thmref{thm:regular_upper}}

To prove \thmref{thm:regular_upper} we consider a symmetric optimum of the ex-ante relaxation, which exists by \lemref{lem:symmetric}. At a symmetric solution the value of the ex-ante relaxation becomes equal to $nqp$, where~$q$ is the probability of selling to each individual buyer and~$p$ is a price that leads to this probability of sale. If~$p$ is used instead as a static price, it produces a revenue of~$p$ times the expected number of goods sold. The latter is given by a binomially distributed random variable with~$n$ trials and success probability~$q$, where in addition the value is capped at the number~$k$ of available goods. To formalize this let $g_{n,k}:[0,1]\to\R_{\geq 0}$ and $f_{n,k}:[0,1]\to\R_{\geq 0}$ such that for all $q\in[0,1]$,
	\begin{align*}
		g_{n,k}(q) &= \E_{X \sim \mathsf{Bin}(n,q)}[ \min\{X,k\}], \\
		 f_{n,k}(q) &= \begin{cases}
 \frac{nq}{\E_{X \sim \mathsf{Bin}(n,q)}[ \min\{X,k\}]}	&\text{if $q>0$ and,} \\
 1 &\text{if $q=0$}.
 \end{cases}
	\end{align*}
We will use the following non-trivial monotonicity result for $g_{n,k}$ and $f_{n,k}$.
\begin{restatable}{lemma}{bernsteinOne}  \label{lem:bernstein1}
	Let $n,k\in\N$ with $k\leq n$. Then $g_{n,k}$ and $f_{n,k}$ are continuous and non-decreasing.
\end{restatable}

To prove the lemma it is useful to view the expectation $\E_{X\sim\mathsf{Bin}(n,q)}[ \min \{X,k\}]$ as an approximation of the function $h(q)=\min\{q,k\}$ by Bernstein polynomials. By the Stone-Weierstrass theorem every continuous function on a closed interval can be uniformly approximated with arbitrary precision by polynomials. The approximation of a function $f$ by its Bernstein polynomial $\Bern_n(f ; \cdot)$ is mathematically well-behaved and in particular preserves properties such as monotonicity or concavity~\citep[\eg][\S~7]{Phil2003}. 
Given a function $f:[0,1]\to\R$ and $n\in\N_{>0}$, the Bernstein polynomial of~$f$ of degree~$n$ is the function $\Bern_n(f; \cdot):[0,1]\to\R$ such that for all $q \in [0,1]$,
\begin{align*}
	\Bern_n(f; q) = \sum_{j=0}^n f\biggl(\frac{j}{n}\biggr) \binom{n}{j} q^j (1-q)^{n-j} .
\end{align*}
Since $h(q) = \min\{q,k/n\}$ is non-decreasing and concave, so is
\begin{align*}
	\Bern_n(h;q) = \sum_{j=0}^{n} \min\{j,q\} \binom{n}{j} q^j (1-q)^{n-j} = \E_{X \sim \Bin(n,q)}[\min\{X,k\}]	= g_{n,k}(q),
\end{align*}
With this knowledge, monotonicity of $f_{n,k}$ can be shown by standard calculus. A formal proof of \lemref{lem:bernstein1} is given in \appref{app:bernstein1}.

\subsection{Proof of \thmref{thm:regular_upper}}

\begin{proof}
By \lemref{lem:symmetric} the ex-ante relaxation has a symmetric optimum, \ie there exists $q^*\in[0,k/n]$ such that $\Rxnk=n\tilde{R}(q^*)$. By regularity of the distribution the revenue curve~$R$ is concave, so $\tilde{R}(q^*)=R(q^*)$. Letting $p^* = \F^{-1}(1-q^*)$, we thus obtain $\Rxnk = nR(q^*)= n q^* \F^{-1}(1-q^*) = n q^* p^*$. On the other hand
\begin{align*}
	\Rsnk(p^*) 
	&= p^* \,\, \E_{X \sim \mathsf{Bin}(n,1-\F(p^*))}\bigl[\min\{X,k\}\bigr] \\
	& \geq p^* \, \min\nolimits_{z\in[q^*\!\!,\,1]} \E_{X \sim \mathsf{Bin}(n,z)}\bigl[\min\{X,k\}\bigr] \\
	&= p^* \,\, \E_{X \sim \mathsf{Bin}(n,q^*)}\bigl[\min\{X,k\}\bigr] ,
\end{align*}
where the inequality holds because $p^* = \max \{p \geq 0 \midd \F(p) \leq 1-q^*\}$ and thus $1-\F(p^*)\in[q^*,1]$, and the second equality by~\lemref{lem:bernstein1}. Thus
\begin{align*}
	\frac{\Rxnk}{\Rsnk} 
	\leq \frac{\Rxnk}{\Rsnk(p^*)}
	\leq \frac{nq^*p^*}{p^*\E_{X \sim \mathsf{Bin}(n, q^*)}[\min\{X,k\}]}
	= \frac{nq^*}{\E_{X \sim \mathsf{Bin}(n, q^*)}[\min\{X,k\}]}.
\end{align*}
By \lemref{lem:bernstein1} the right-hand side of this expression is non-decreasing in $q^*$, and since $q^*\leq k/n$ we conclude that
\begin{align*}
	\frac{\Rxnk}{\Rsnk}	& \leq \frac{k}{\E_{X \sim \mathsf{Bin}(n, k/n)}[\min\{X,k\}]},
\end{align*}
as claimed.
\end{proof}

\subsection{A Matching Lower Bound}

To obtain a lower bound matching the upper bound of \thmref{thm:regular_upper} we consider a family of distributions parameterized by $r\in[0,1]$ and $\epsilon>0$ that are regular and whose revenue curves are piece-wise linear and bend at $\epsilon$. The limiting revenue curve~$R$ as $\epsilon\to 0$ is linear on $(0,1]$ and satisfies $R(0) = 0$, $R(1)=1$ and $\lim_{q\downarrow 0}R(q)=r$. Examples of these distributions and their revenue curves are shown in \figref{fig:regular_lower}.
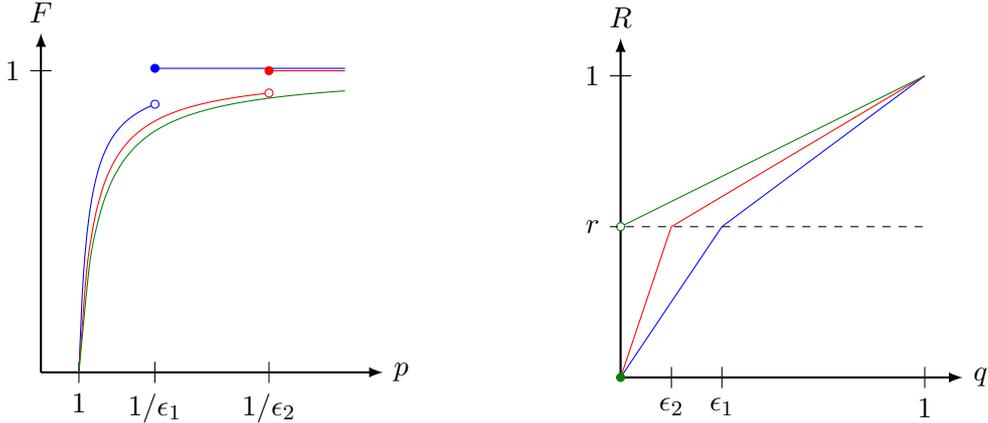
\begin{figure}[tb]
\centering
\begin{subfigure}{0.4\textwidth}
\begin{tikzpicture}[yscale=4,xscale=0.5]
	\draw[-latex,thick] (0,0) -- (9,0) node[right]{$p$};
	\draw (0,1) node[left,xshift=-4pt] {$1$};
	\draw (-8pt,1) -- (8pt,1);
	\draw[-latex,thick] (0,0) -- (0,1.125) node[above] {$F$};
	\draw (1,-1pt) -- (1,1pt);
	\draw (1,0) node[below,yshift=-4pt]{$1$};
	\draw (3,-1pt) -- (3,1pt);
	\draw (3,0) node[below,yshift=-4pt]{$1/\epsilon_1$};
	\draw (6,-1pt) -- (6,1pt);
	\draw (6,0) node[below,yshift=-4pt]{$1/\epsilon_2$};
	\draw[domain=1:3,smooth,variable=\p,blue] plot ({\p},{1 - (1 - 0.5/0.667)/(\p - 0.5/0.667)});
	\draw (3,0.889) node[draw=blue,fill=white,inner sep=0pt,circle,minimum size=3pt] {};
	\draw (3,1.008) node[draw=blue,fill=blue,inner sep=0pt,circle,minimum size=3pt] {};
	\draw[domain=3:8,variable=\p,blue] plot ({\p}, {1.008});
	\draw[domain=1:6,smooth,variable=\p,red] plot ({\p},{1 - (1 - 0.5/0.833)/(\p - 0.5/0.833)});
	\draw (6,0.926) node[draw=red,fill=white,inner sep=0pt,circle,minimum size=3pt] {};
	\draw (6,1) node[draw=red,fill=red,inner sep=0pt,circle,minimum size=3pt] {};
	\draw[domain=6:8,variable=\p,red] plot ({\p}, {1});
	\draw[domain=1:8,smooth,variable=\p,green!50!black] plot ({\p},{1 - (1 - 0.5/1)/(\p - 0.5/1)});
 \end{tikzpicture}
 \end{subfigure}  \hspace{1cm}
 \begin{subfigure}{0.4\textwidth}
 \begin{tikzpicture}[yscale=4,xscale=4]
       \draw[-latex,thick] (0,0) -- (1.125,0) node[right]{$q$};
       \draw (-1pt,1) -- (1pt,1);
       \draw (0,1) node[left,xshift=-4pt]{$1$};
 			\draw[-latex,thick] (0,0) -- (0,1.125) node[above] {$R$};
 		\draw[domain=0:0.333,variable=\q,blue] plot ({\q},{\q * 0.5/0.333});
 		\draw[domain=0.333:1,variable=\q,blue] plot ({\q},{\q * 0.5/0.667 + (1 - 0.5/0.667)});
 		\draw[domain=0:0.167,variable=\q,red] plot ({\q},{\q * 0.5/0.167});
 		\draw[domain=0.167:1,variable=\q,red] plot ({\q},{\q * 0.5/0.833 + (1 - 0.5/0.833)});
 		\draw[domain=0:1,variable=\q,green!50!black] plot ({\q},{\q * 0.5/1 + (1 - 0.5/1)});
 		\draw (0,0) node[draw=green!50!black,fill=green!50!black,inner sep=0pt,circle,minimum size=3pt] {};
 		\draw (0,0.5) node[draw=green!50!black,fill=white,inner sep=0pt,circle,minimum size=3pt] {};
 		\draw[dashed] (-1pt,0.5) -- (1,0.5);
 		\draw (0,0.5) node[left,xshift=-4pt]{$r$};
 		\draw (0.333,-1pt) -- (0.333,1pt);
 		\draw (0.333,0) node[below,yshift=-4pt] {$\epsilon_1$};
 		\draw (0.167,-1pt) -- (0.167,1pt);
 		\draw (0.167,0) node[below,yshift=-4pt] {$\epsilon_2$};
 		\draw (1,-1pt) -- (1,1pt);
 		\draw (1,0) node[below,yshift=-4pt]{$1$};
 	\end{tikzpicture}
 \end{subfigure}
 \caption{Cumulative distribution functions (left) and revenue curves (right) for probability distributions used in the proof of \thmref{thm:regular_lower}. The blue and red curves correspond to the cases where $r=1/2$ and where $\epsilon=\epsilon_1=1/3$ and $\epsilon=\epsilon_2=1/6$, respectively. The green curve corresponds to the limit as $\epsilon\to 0$.}
 \label{fig:regular_lower}
 \end{figure}

Using dynamic prices it is possible to sell a good with probability close to~$0$ to each of the first $n-k$ buyers and with probability close to~$1$ to each of the remaining~$k$ buyers, thus obtaining an overall revenue close to $(n-k)r+k$. A static price that sells with probability~$q$ to an individual buyer obtains revenue at most $r+q(1-r)$ for each good sold, and the number of such goods is $\E_{X\sim\mathsf{Bin}(n,q)}\min\{X,k\}$. Our goal now is to choose~$r$ such that for the corresponding optimal choice of~$q$ the gap in revenue between dynamic and static prices is as large as possible. This maximin problem turns out to be rather difficult to solve directly, but we will appeal to Berge's maximum principle \citep{Berge1963} to show the existence of a value of $r$ that leads to an optimal choice of $q=k/n$. By the upper bound of \thmref{thm:regular_upper} and by the existence of a symmetric optimum of the ex-ante relaxation the latter is in fact necessary for tightness.

We obtain the following result, which is proved formally in \appref{app:regular_lower}.
\begin{restatable}{theorem}{regularLower}  \label{thm:regular_lower}
Let $n,k\in\N$ with $k\leq n$ and $\delta>0$. Then there exists a regular distribution~$F$ such that
\[
	\frac{\Rdnk}{\Rsnk} \geq  \frac{k}{\E_{X \sim \mathsf{Bin}(n, k/n)}[\min\{X,k\}]} - \delta .
\]
\end{restatable}

\section{General Distributions}
\label{sec:general}

We proceed by bounding the revenue gap between dynamic and static prices also for general distributions. The following upper bound, which will again turn out to be tight for all values of $n$ and $k$, reveals a qualitative difference between the gap for distributions that are regular and the gap for those that are not.
\begin{theorem}  \label{thm:general-upper}
	For any number~$n$ of buyers with values drawn independently from a distribution~$F$ and any number $k\leq n$ of goods,
\[
	\frac{\Rxnk}{\Rsnk} \leq 2-\frac{k}{n}.
\]
\end{theorem}

To readers familiar with the ironing technique of \cite{Myer81a}, which in the context of incentive-compatible mechanisms enables a uniform treatment of regular and irregular distributions, it may seem surprising at first that the two types of distributions would give rise to different bounds. Before proving \thmref{thm:general-upper} we therefore demonstrate that a straightforward application of the ironing technique can bound the revenue from dynamic prices only relative to that from a \emph{static lottery} over at most two prices and not relative to that from a single static price. Such a lottery can be defined formally in terms of $p_1,p_2\in\R$ and $\beta\in [0,1]$ and offers a good to each buyer at price $p_1$ with probability $\beta$ and at price $p_2$ with probability $1-\beta$. Denoting by $\Rl_{F,n,k}$ the maximum revenue that can be achieved by a static lottery over two prices for~$n$ buyers and~$k$ goods, it can be shown with the same arguments as in \secref{sec:regular} that for any $n,k\in\N$ with $k\leq n$ and even for irregular distributions,
\[
	\frac{\Rxnk}{\Rl_{F,n,k}} \leq \frac{k}{\E_{X \sim \Bin(n,k/n)}[\min\{X,k\}]}	.
\]
The crucial observation is that in the irregular case, a symmetric solution of the ex-ante relaxation that sells to each buyer with probability~$q$ may not correspond to a single price but rather to a lottery over two prices with an overall probability of sale equal to~$q$. Since static prices and static lotteries over prices are equally powerful for regular distributions, and in particular for the family of distributions used in the proof of \thmref{thm:regular_lower}, the bound is tight for all values of~$n$ and~$k$. A straightforward derandomization of a static lottery over two prices using the method of conditional expectations does not lead to a bound like that of \thmref{thm:general-upper} because it may offer \emph{two distinct prices} to different subsets of the buyers. \thmref{thm:general-upper}, by contrast, holds relative to a single static price.

\subsection{Prerequisites for the Proof of \thmref{thm:general-upper}}

To prove \thmref{thm:general-upper} we again start with a symmetric optimum of the ex-ante relaxation, which is guaranteed to exist by \lemref{lem:symmetric}. In \secref{sec:regular} we were ultimately able to express the revenue gap between a symmetric optimum that sells to each buyer with probability~$q$ in terms of the function
\begin{align*}
	f_{n,k}(q) = \begin{cases}
		\frac{nq}{\E_{X \sim \Bin(n,q)}[\min\{X,k\}]} & \text{if $q>0$ and} \\
		1 & \text{if $q=0$}.
	\end{cases}
\end{align*}
We then established monotonicity of $f_{n,k}$ and concluded that the revenue gap is maximized when $q=k/n$.

The reason why this argument does not carry over to irregular distributions is that for such distributions a symmetric solution of the ex-ante relaxation may not correspond to a single static price~$p$ but rather to a static lottery over prices $p_1$ and $p_2$. To nevertheless prove an upper bound on the revenue gap in this case we will instead consider a carefully chosen lottery over the static prices~$p_1$ and~$p_2$, \ie a lottery that selects one of these prices randomly and then offers it to all of the bidders. Our first objective will be to show that the lottery can be chosen in such a way that its revenue gap relative to the ex-ante relaxation is bounded by a convex combination of the form
\begin{align*}
	\lambda_1 f_{n,k}(q_1) + \lambda_2 f_{n,k}(q_2)
\end{align*}
under the constraint that $\lambda_1q_1+\lambda_2q_2=k/n$. We will then argue that this gap is maximized when $q_1=0$ and $q_2=1$. Together with monotonicity of~$f_{n,k}$, the latter would follow from its convexity, but a proof of convexity unfortunately remains elusive. We will show instead that $f_{n,k}$ lies below its secant line, \ie below the function $s_{n,k}:[0,1]\to\R_{\geq 0}$ such that for all $q \in [0,1]$,
\begin{align*}
	s_{n,k}(q) &= (1-q)f_{n,k}(0) + qf_{n,k}(1) .
\end{align*}
This is established by the following lemma.

\begin{restatable}{lemma}{lemBernsteinTwo}
\label{lem:bernstein2}
	Let $n,k\in \N$ such that $k\leq n$. Then, for all $q \in [0,1]$, $f_{n,k}(q)\leq s_{n,k}(q)$.
\end{restatable}

To prove the lemma we again exploit the close connection between $f_{n,k}$ and approximations by Bernstein polynomials. For $\alpha>0$ let $h_{\alpha}(q)=\min\{q,\alpha\}$, and observe that for \lemref{lem:bernstein2} it suffices to show that for all $\alpha>0$ and all $n\geq\lceil 1/\alpha\rceil$,
\begin{align*}
	\frac{q}{\Bern_n(h_{\alpha}; q)} \leq (1-q) + \frac{q}{\alpha}.
\end{align*}
Indeed, the lemma follows from this claim by setting $\alpha=k/n$. For $n=\lceil 1/\alpha\rceil$ the Bernstein polynomial $\Bern_n(h_\alpha;q)$ has a very simple structure since $\min\{j/n,\alpha\}=\alpha$ for all $j\geq 1$, and the claim can be shown using a Bernoulli-type inequality. To show that the claim holds for $n\geq\lceil 1/\alpha \rceil$ we use that the approximation of a concave function~$h$ by Bernstein polynomials is monotone in~$n$ in the sense that for all $q\in[0,1]$, $\Bern_{n}(h;q)\leq\Bern_{n+1}(h;q)\leq h(q)$. A detailed proof of \lemref{lem:bernstein2} is given in \appref{app:bernstein2}. 

\subsection{Proof of \thmref{thm:general-upper}}

\begin{proof}
By \lemref{lem:symmetric} and by definition of~$\tR$ there respectively exist $q^*\in[0,k/n]$ such that $\smash[b]{\Rxnk}=n\tilde{R}(q^*)$ and $q_1,q_2,\beta\in[0,1]$ such that $\tilde{R}(q^*)=\beta R(q_1)+(1-\beta)R(q_2)$. 
Thus
\begin{align*}
	\Rxnk 
	= n\bigl(\beta R(q_1) + (1-\beta)R(q_2) \bigr) .
\end{align*}
Since $q^*\leq k/n$, it further holds that  $\beta q_1+(1-\beta)q_2\leq k/n$.

Let 
\begin{align*}
	\alpha_1 &= \frac{\beta R(q_1)}{\Rsnk(\F^{-1}(1-q_1))} & 	& \text{and} & \alpha_2 &= \frac{(1-\beta) R(q_2)}{\Rsnk(\F^{-1}(1-q_2))},
\end{align*}
and consider a lottery that selects a static price of~$p_1 = \F^{-1}(1-q_1)$ with probability proportional to~$\alpha_1$ and a static price of~$p_2 = \F^{-1}(1-q_2)$ with probability proportional to~$\alpha_2$.
This lottery achieves an expected revenue of
\begin{align*}
	\frac{\alpha_1}{\alpha_1+\alpha_2}	\Rsnk(p_1) + \frac{\alpha_2}{\alpha_1+\alpha_2}	\Rsnk(p_2) = 
	\frac{\beta R(q_1)}{\alpha_1+\alpha_2} + \frac{(1-\beta)R(q_2)}{\alpha_1+\alpha_2} = \frac{\Rxnk}{n(\alpha_1+\alpha_2)} .
\end{align*}
Since a lottery over the static prices~$p_1$ and~$p_2$ achieves a revenue of $\Rxnk/n(\alpha_1 + \alpha_2)$, so does one of the static prices, \ie either $\Rxnk/\Rsnk(p_1)\leq n(\alpha_1 + \alpha_2)$ or $\Rxnk/\Rsnk(p_2)\leq n(\alpha_1 + \alpha_2)$. In particular we obtain
\begin{multline*}
	\frac{\Rxnk}{\Rsnk} \leq n(\alpha_1 + \alpha_2) 
	\leq \\
	\sup\Biggl\{ \frac{\beta nR(q_1)}{\Rsnk(\F^{-1}(1-q_1))} + 	\frac{(1-\beta)nR(q_2)}{\Rsnk(\F^{-1}(1-q_2))} \midd \beta,q_1,q_2\in [0,1], \beta q_1 + (1-\beta)q_2\leq\frac{k}{n} \Biggr\} .
\end{multline*}
Now, for $q\in\{q_1,q_2\}$,
\begin{align*}
	\frac{n R(q)}{\Rsnk(\F^{-1}(1-q))} 
	&= \frac{n q \F^{-1}(1-q)}{\F^{-1}(1-q)\E_{X \sim \mathsf{Bin}(n,q)}[\min\{X,k\}]} %
	= \frac{n q}{\E_{X \sim \mathsf{Bin}(n,q)}[\min\{X,k\}]} = f_{n,k}(q) ,
\end{align*}
and thus
\begin{align*}
	\frac{\Rxnk}{\Rsnk} &\leq \sup \Bigl\{ \beta f_{n,k}(q_1) + (1-\beta) f_{n,k}(q_2) \midd \beta,q_1,q_2\in[0,1], \beta q_1 + (1-\beta)q_2\leq\frac{k}{n} \Bigr\}.
\end{align*}

We may assume without loss of generality that $q_1\leq q_2$. Since, by \lemref{lem:bernstein1}, $f_{n,k}(q)$ is non-decreasing in $q$, we may further increase $q_2$ until the constraint $\beta q_1+(1- \beta)q_2\leq k/n$ becomes tight and therefore assume without of loss of generality that $q_1\leq k/n$, $q_2>k/n$, and $\beta q_1 + (1-\beta)q_2=k/n$. Thus
\begin{align*}
	\frac{\Rxnk}{\Rsnk} &\leq \sup \Biggl\{ \frac{q_2-\frac{k}{n}}{q_2-q_1}f_{n,k}(q_1) + \frac{\frac{k}{n} - q_1}{q_2-q_1} f_{n,k}(q_2)\midd q_1 \in \Bigl[0,\frac{k}{n}\Bigr], q_2 \in \Bigl(\frac{k}{n},1\Bigr] \Biggr\}.
\end{align*}
Geometrically the expression in the supremum corresponds to the value at $k/n$ of the secant line of the function $f_{n,k}$ through the points $(q_1, f_{n,k}(q_1))$ and $(q_2, f_{n,k}(q_2))$. By \lemref{lem:bernstein1} $f_{n,k}$ is non-decreasing and by \lemref{lem:bernstein2} it lies below the secant line through the points $(0,f_{n,k}(0))$ and $(1,f_{n,k}(1))$, so the supremum is attained when $q_1=0$ and $q_2=1$. Since $f_{n,k}(0)=1$ and $f_{n,k}(1)=n/k$,
\begin{align*}
\frac{\Rxnk}{\Rsnk} &\leq	1-\frac{k}{n} + 1 = 2- \frac{k}{n} 
\end{align*}
as claimed.
\end{proof}

\subsection{A Matching Lower Bound}

In search of a lower bound matching the upper bound of \thmref{thm:general-upper} we may restrict our attention to distributions where the optimal static price extracts revenue~$1$ from each of~$k$ buyers. The proof of \thmref{thm:general-upper} suggests, on the other hand, that the optimal dynamic prices should sell with vanishing probability to $n-k$ of the buyers while still extracting revenue $k/n$ from each of them in expectation, and with probability and expected revenue approaching~$1$ for each of the remaining~$k$ buyers. We show in \appref{app:lower} that this can be achieved for a discrete distribution with support of size two, and obtain the following result.
\begin{restatable}{theorem}{irregularLower}  \label{thm:irregular-lower}
	Let $n,k\in\N$ with $k\leq n$ and $\delta>0$. Then there exists a distribution~$F$ such that
\[
	\frac{\Rdnk}{\Rsnk} \geq 2-\frac{k}{n}-\delta .
\]
\end{restatable}

\appendix

\section{Auxiliary Lemmas for \secref{sec:regular}}
\label{app:regular}

\subsection{Proof of \corref{cor:poisson}}
\label{app:poisson}

\corPoisson*

To prove \corref{cor:poisson}, we need the following auxiliary lemma.
\begin{lemma}  \label{lem:poisson_monotone}
	Let $n,k\in\N$ such that $k\leq n$. Then
\begin{align*}
	\E_{X \sim \mathsf{Bin}(n,k/n)}[\min\{X,k\}] \geq \E_{X \sim \mathsf{Poi}(k)} [\min\{X,k\}].
\end{align*}
\end{lemma}
\begin{proof}
	For $D\in\{\mathsf{Bin}(n,k/n),\mathsf{Poi}(k)\}$,
\begin{align*}  %
	\E_{X\sim D}[\min\{X,k\}] 
	&= \sum_{j=1}^{k} \P_{X\sim D}[X \geq j]
	=\sum_{j=0}^{k-1} \Bigl(1-\P_{X\sim D}[X \leq j]\Bigr).  
\end{align*}
	By a result of \citet{AndersonS67} $\P_{\mathsf{Bin}(n,k/n)}(X \leq j)\leq\P_{\mathsf{Poi}(k)}(X \leq j)$ precisely when $0\leq j\leq k-k/(n+1)$. Since $k\leq n$, this condition is satisfied for all terms in the sum on the right-hand side, and thus
\begin{align*}
	\E_{X \sim \mathsf{Bin}(n,k/n)}[\min\{X,k\}] 
	&=\sum_{j=0}^{k-1} \Bigl( 1-\P_{\mathsf{Bin}(n,k/n)}[X \leq j] \Bigr) \\
	&\geq \sum_{j=0}^{k-1} \Bigl( 1-\P_{\mathsf{Poi}(k)}[X \leq j] \Bigr) \\
	&= \E_{X \sim \mathsf{Poi}(k)}[\min\{X,k\}],
\end{align*}
as claimed. 
\end{proof}

\begin{proof}[Proof of \corref{cor:poisson}]
We claim that 
\begin{align*}
	\frac{\Rxnk}{\Rsnk} 
	&\leq \frac{k}{\E_{X \sim \Poi(k)}[\min\{X,k\}]} 
	= \frac{1}{1 - \frac{k^k}{e^{k}k!}} .
\end{align*}
Indeed the inequality holds by \lemref{lem:poisson_monotone}, and the equality because
\begin{align*}
	\E_{X \sim \mathsf{Poi}(k)}[\min\{X,k\}] &= \sum_{j=0}^{\infty} \frac{k^j}{e^{k}j!} \min\{j,k\} \\
	&= k - \sum_{j=0}^{k-1} \frac{k^j}{e^{k}j!} \bigl(k-j\bigr)\\
	&= k - \sum_{j=0}^{k-1} \frac{k^{j+1}}{e^{k}j!} + \sum_{j=1}^{k-1} \frac{k^{j}}{e^{k}(j-1)!}\\
	&= k - \sum_{j=0}^{k-1} \frac{k^{j+1}}{e^{k}j!} + \sum_{j=0}^{k-2} \frac{k^{j+1}}{e^{k}j!} = k - \frac{k^k}{e^{k}(k-1)!} .   
\end{align*}
\end{proof}

\subsection{Proof of \lemref{lem:bernstein1}}
\label{app:bernstein1}

Recall that
\begin{align*}
	g_{n,k}(q) &= \E_{X \sim \mathsf{Bin}(n,q)}[ \min\{X,k\}], \\
	f_{n,k}(q) &= \begin{cases}
		\frac{nq}{\E_{X \sim \mathsf{Bin}(n,q)}[ \min\{X,k\}]}	&\text{if $q>0$ and,} \\
 	 	1 &\text{if $q=0$}.
 \end{cases}
\end{align*}

\bernsteinOne*

\begin{proof}
	Continuity of $g_{n,k}$ on $[0,1]$ and of $f_{n,k}$ on $(0,1]$ is obvious. For continuity of $f_{n,k}$ at $q=0$ note that
\begin{align*}
	\lim_{q \to 0}f_{n,k}(q) 
	&= \lim_{q \to 0} \frac{nq}{\sum_{j=1}^{n} \min\{j,k\} \binom{n}{j} q^{j} (1-q)^{n-j}} \\
	&= \lim_{q \to 0} \frac{1}{\sum_{j=1}^{n} \min\bigl\{\frac{j}{n},\frac{k}{n}\bigr\} \binom{n}{j} q^{j-1} (1-q)^{n-j}} 
	= \frac{1}{\min\{\frac{1}{n},\frac{k}{n}\}\binom{n}{1}} = 1 ,
\end{align*}
where the third inequality holds because $\lim_{q\to 0}q^{j-1}=1$ if $j=1$ and $\lim_{q\to 0}q^{j-1}=0$ if $j>1$.

To see that $g_{n,k}$ is non-decreasing, recall the definition of the Bernstein polynomial of degree  $n\in\N_{>0}$ of a function $f:[0,1]\to\R$ as the function $\Bern_n(f;\cdot):[0,1]\to\R$ such that for all $q\in[0,1]$,
\begin{align*}
	\Bern_n(f; q) = \sum_{j=0}^n f\biggl(\frac{j}{n}\biggr) \binom{n}{j} q^j (1-q)^{n-j} .
\end{align*}
Let $h:[0,1]\to\R$ such that $h(x)=\min\{x,k/n\}$, and observe that $g_{n,k}(q)=n\Bern_n(h;q)$. Since~$h$ is non-decreasing and concave and approximation by Bernstein polynomials preserves these properties~\citep[\S~7]{Phil2003}, $g_{n,k}=n\Bern_n(h; q)$ is non-decreasing and concave.

By continuity of $f_{n,k}$ it suffices to show that $f_{n,k}$ is non-decreasing on $(0,1)$. For $q>0$, $f_{n,k}(q) = \frac{nq}{g_{n,k}(q)}$. Since
\begin{align*}
	g_{n,k}(q) = n\sum_{j=0}^n \min\biggl\{\frac{j}{n},\frac{k}{n}\biggr\} \binom{n}{j} q^j (1-q)^{n-j}	,
\end{align*}
both $g_{n,k}$ and $f_{n,k}$ are differentiable on $(0,1)$. In particular, for $q\in(0,1)$,
\begin{align*}  %
	f_{n,k}'(q) = n\frac{g_{n,k}(q) - qg_{n,k}'(q)}{g_{n,k}(q)^2}.
\end{align*}
Fix $q\in(0,1)$. Since $g_{n,k}(0)=0$ and by the mean value theorem, there exists some $\xi\in(0,q)$ such that $g'(\xi)=g_{n,k}(q)/q$. Together with concavity of~$g_{n,k}$ this implies that $g_{n,k}(q)=qg'_{n,k}(\xi)\geq qg'_{n,k}(q)$, so $f_{n,k}'$ is non-negative and $f_{n,k}$ non-decreasing on $(0,1)$. 
\end{proof}

\subsection{Proof of \thmref{thm:regular_lower}}
\label{app:regular_lower}

\regularLower*
\begin{proof}
	Let $r\in\R$ and $\epsilon>0$, and consider a distribution with support $[1,\frac{1}{\epsilon}]$ and cumulative distribution function $F$ such that
	\begin{align*}
		F(p) = \begin{cases}
		1 - \frac{1-\frac{1-r}{1-\epsilon}}{p-\frac{1-r}{1-\epsilon}} & \text{if $p\in[1,\frac{1}{\epsilon})$ and} \\
		1 & \text{if $p=\frac{1}{\epsilon}$.}
		\end{cases}
	\end{align*}
	This distribution is continuous on $[1,\frac{1}{\epsilon})$ and has point mass at $\frac{1}{\epsilon}$. Its revenue curve $R$ is piecewise linear with slope $r/\epsilon$ on $(0,\epsilon)$ and slope $\frac{1-r}{1-\epsilon}$ on $(\epsilon,1)$, and is given by
	\begin{equation*}
		R(q) = \begin{cases}
			\frac{r}{\epsilon} \cdot q  &\text{if $q\in[0,\epsilon]$ and} \\
		 	\frac{1-r}{1-\epsilon} \cdot q + (1 - \frac{1-r}{1-\epsilon}) & \text{if $q\in(\epsilon,1)$.}
		\end{cases}
	\end{equation*}

Consider a matrix $\vec p\in\R^{n\times k}$ of dynamic prices such that
\begin{align*}
p_{i,j} &= \F^{-1}(1-\epsilon)= \frac{R(\epsilon)}{\epsilon} = \frac{r}{\epsilon} & &\text{for all $i\in\{1,\dots,n-k\}$, $j\in [k]$}, \\
p_{i,j} &= \F^{-1}(0)=R(1)=1 & &\text{for all $i\in\{n-k+1,\dots,n\}$, $j\in[k]$.}
\end{align*}
The revenue obtained by these prices is at least
\begin{align*}  %
	\sum_{i=1}^{n-k} \bigl((1-\epsilon)^{i-1}R(\epsilon) \bigr)+k(1-\epsilon)^{n-k} R(1) 
	\geq \bigl( (n-k)r+k \bigr) (1-\epsilon)^{n-k},
\end{align*}
because each of the first $n-k$ buyers contributes $R(\epsilon)$ to the revenue at least in the case where no goods have been sold to any of the earlier buyers, and each of the last $k$ buyers contributes $R(1)$ at least in the case where no goods were sold to the first $n-k$ buyers.

In determining the maximum revenue from a static price we can restrict our attention to prices $p\in[1,\frac{1}{\epsilon}]$, as prices below $1$ are dominated in terms of revenue by a price of~$1$ and prices above $\frac{1}{\epsilon}$ sell with probability zero.  Prices $p\in[1,\frac{1}{\epsilon}]$ allow us to sell to a single buyer with probability $q$ for any $q\in[\epsilon,1]$, and obtain an overall revenue of
\begin{align*}  %
	\E_{X \sim \textsf{Bin}(n,q)}[\min\{X,k\}] \frac{R(q)}{q}
	&\leq \E_{X \sim \textsf{Bin}(n,q)}[\min\{X,k\}] \frac{r+q(1-r)}{q} ,
\end{align*}
where the inequality holds because for any $\epsilon>0$, $R(q)\leq r+q(1-r)$. 

For $n,k\in\N$ with $k\leq n$ and $\epsilon>0$, define $T_{n,k,\epsilon}:[\epsilon,1]\times[0,1]\to\R$ such that for all $q\in[\epsilon,1]$ and $r\in[0,1]$,
\begin{align*}
	T_{n,k,\epsilon}(q,r)
	&= %
	 \frac{\bigl( (n-k)r+k \bigr)q}{\E_{X \sim \textsf{Bin}(n,q)}[\min\{X,k\}] \bigl(r+q(1-r)\bigr)} \,(1-\epsilon)^{n-k}.
\end{align*}
Note that $T_{n,k,\epsilon}(q,r)$ is well-defined, since the denominator is greater than zero for all $q\geq\epsilon>0$ and $k\geq 1$, and continuous on its domain. Note further that for all $\epsilon>0$
\[
	\frac{\Rdnk}{\Rsnk} \geq \max\nolimits_{r\in[0,1]}\min\nolimits_{q\in [\epsilon,1]} T_{n,k,\epsilon}(q,r) ,
\]
if indeed the maximin problem has a solution.

The maximin problem turns out to be rather difficult to solve directly, and we will instead give a non-constructive proof of the claimed lower bound. To this end, we first show the existence of $r^*\in[0,1]$ such that $\min_{q\in [0,1]} T_{n,k,\epsilon}(q,r^*)$ is attained at $q=k/n$, \ie that $k/n\in\arg\min_{q \in [0,1]}T_{n,k,\epsilon}(q,r^*)$. For $n,k\in\N$ with $k\leq n$ and $\epsilon>0$, let $\T_{n,k,\epsilon}:[0,1]\to 2^{[\epsilon,1]}$ such that for all $r\in[0,1]$,
\[
	\T_{n,k,\epsilon}(r) = \arg\min\nolimits_{q \in [\epsilon,1]} T_{n,k}(q,r) .
\]
By Berge's maximum theorem \citep{Berge1963} this correspondence is nonempty- and compact-valued and upper hemicontinuous, which implies in particular that the graph 
\[
	G(\T_{n,k,\epsilon}) =  \{(q,r) : r \in [0,1], q \in \T_{n,k,\epsilon}\}
\]
is a closed and connected set in $\R^2$. Note further that 
\begin{align*}
	\T_{n,k,\epsilon}(0) &= \arg\min\nolimits_{q \in [\epsilon,1]} T_{n,k,\epsilon}(q,0) \\
	&= \arg\min\nolimits_{q \in [\epsilon,1]} \biggl\{ \frac{k}{\E_{X \sim \textsf{Bin}(n,q)}[\min\{X,k\}]} (1-\epsilon)^{n-k} \biggr\}=\{1\} ,
\end{align*}
where the last equality holds because $\frac{k}{\E_{X \sim \Bin(n,q)}[\min\{X,k\}]}$ is strictly decreasing in $q$,
and that
\begin{align*}
	\T_{n,k,\epsilon}(1) &= \arg\min\nolimits_{q \in [\epsilon,1]} T_{n,k,\epsilon}(q,1) \\
	&=\arg\min\nolimits_{q \in [\epsilon,1]} \biggl\{ \frac{nq}{\E_{X \sim \textsf{Bin}(n,q)}[\min\{X,k\}]}(1-\epsilon)^{n-k} \biggr\} = \{0\} ,
\end{align*}
where the last equality holds because for all $q>0$, 
\[
	\frac{nq}{\E_{X \sim \textsf{Bin}(n,q)}[\min\{X,k\}]} 
	= \frac{\E_{X \sim\textsf{Bin}(n,q)}[X]}{\E_{X \sim \textsf{Bin}(n,q)}[\min\{X,k\}]}>1 .
\]
Since $G(\T_{n,k,\epsilon})$ is a closed and connected set in $\R^2$ and contains the points $(0,1)$ and $(1,0)$, it must have a nonempty intersection with the line $\{(k/n,r):r \in [0,1]\}$. This intersection yields $r^*\in[0,1]$ such that $k/n\in\T_{n,k,\epsilon}(r^*)$.

Now 
\begingroup 
\begin{align*}
	\max\nolimits_{r \in [0,1]} \min\nolimits_{q \in [\epsilon,1]} T_{n,k,\epsilon}(q,r) &\geq \min\nolimits_{q \in [\epsilon,1]} \Bigl\{ T_{n,k,\epsilon}(q,r^*) \Bigr\}\\
	&= \min\nolimits_{q\in[0,1]} \biggl\{\frac{\bigl((n-k)r^*+k \bigr) q}{\E_{X \sim \textsf{Bin}(n,q)}[\min\{X,k\}] \bigl(r^*+q(1-r^*)\bigr)} (1-\epsilon)^{n-k}\biggr\} \\
	&= \frac{\bigl( (n-k)r^*+k \bigr)\frac{k}{n}}{\E_{X \sim \textsf{Bin}(n,q)}[\min\{X,k\}] \bigl(r^*+\frac{k}{n}(1-r^*)\bigr)}(1-\epsilon)^{n-k} \\
	&= \frac{\bigl( nr^*+(1-r^*)k \bigr)k}{\E_{X \sim \textsf{Bin}(n,q)}[\min\{X,k\}] \bigl(nr^*+(1-r^*)k\bigr)}(1-\epsilon)^{n-k}\\
	&= \frac{k}{\E_{X \sim \textsf{Bin}(n,q)}[\min\{X,k\}]}(1-\epsilon)^{n-k}.
\end{align*}
\endgroup
By continuity of $T_{n,k,\epsilon}$ in~$\epsilon$ and by continuity of the maximum and minimum,
\[
	\max\nolimits_{r \in [0,1]} \min\nolimits_{q \in [\epsilon,1]} \lim\nolimits_{\epsilon\to 0}T_{n,k,\epsilon}(q,r)
	= \lim\nolimits_{\epsilon\to 0} \max\nolimits_{r \in [0,1]} \min\nolimits_{q \in [\epsilon,1]} T_{n,k,\epsilon}(q,r),
\]
so choosing $\epsilon$ small enough establishes the claim. 
\end{proof}

\section{Auxiliary Lemmas for \secref{sec:general}}

\subsection{Proof of \lemref{lem:bernstein2}}
\label{app:bernstein2}

Recall that
\begin{align*}
	f_{n,k}(q) &= \begin{cases}
		\frac{nq}{\E_{X \sim \mathsf{Bin}(n,q)}[\min\{X,k\}]}	&\text{if $q>0$,} \\
 	 	1 &\text{if $q=0$, and}
 \end{cases} \\
	s_{n,k}(q) &= (1-q)f_{n,k}(0) + qf_{n,k}(1) .
\end{align*}

\lemBernsteinTwo*

The proof uses the following variant of Bernoulli's inequality.
\begin{lemma}
\label{lem:bernoulli}
	Let $k\in\N \cup \{0\}$ and $x\in(0,1)$. Then $(1-x)^k \leq (1+kx)^{-1}$.
\end{lemma}
\begin{proof}
	For $k=0$, the statement is trivial. For $k \geq 1$, we show the statement by induction. For $k=1$ the claim is that $1-x\leq(1+x)^{-1}$, which holds when $x>0$. Then, for $k\geq 1$,
	\begin{align*}
		(1-x)^{k+1} &= (1-x)^k(1-x) \\
		&\leq (1+kx)^{-1}(1+x)^{-1}\\
		&= (1+kx+x+kx^2)^{-1} \\
		&\leq \bigl(1+(k+1)x\bigr)^{-1},
	\end{align*}
	where the first inequality uses the induction hypothesis, the second inequality that $kx^2\geq 0$, and both inequalities that $x\geq 0$ and $1-x\geq 0$. 
\end{proof}

\begin{proof}[Proof of \lemref{lem:bernstein2}]
The lemma holds trivially when $k=n$ since $f_{n,n}(q)=1$ for all $q\in[0,1]$. We proceed to show it for the case where $k<n$.

For $\alpha\in(0,1)$ let $h_{\alpha}:[0,1]\to\R$ and $s_{\alpha}:[0,1]\to\R$ such that $h_{\alpha}(q)=\min\{q,\alpha\}$ and $s_{\alpha}(q)=(1-q)+\frac{q}{\alpha}$ for all $q\in[0,1]$. We will show the lemma by showing the stronger claim that for all $\alpha\in(0,1)$, $n\in\N$ with $n\geq\lceil 1/\alpha\rceil$, and $q\in(0,1)$,
\[
	\frac{q}{\Bern_n(h_\alpha; q)}\leq s_{\alpha}(q) .
\]
This claim indeed implies the lemma because for all $n,k\in\N$ with $k<n$, $q\in(0,1)$, and $\alpha=k/n$,
\[
	f_{n,k}(q) = \frac{q}{\Bern_n(h_{k/n}; q)} = \frac{q}{\Bern_n(h_\alpha; q)} \leq s_{\alpha}(q) = s_{n,k}(q) ,
\]
and thus $f_{n,k}(q)\leq s_{n,k}(q)$ for all $q\in[0,1]$ by continuity of $s_{n,k}$ and $f_{n,k}$.

We first show the claim for $n=\lceil 1/\alpha\rceil$, in which case
\begin{align*}
	\frac{1}{n} = \frac{1}{\lceil 1/\alpha \rceil} \leq \frac{1}{1/\alpha} = \alpha  \qquad \text{and} \\ 
	\frac{2}{n} = \frac{2}{\lceil 1/\alpha \rceil} \geq \frac{2}{1/\alpha+1} = \frac{2\alpha}{1 + \alpha} \geq \alpha.
\end{align*}
Then
\begin{align*}
	\Bern_n(h_{\alpha}; q) 
	&= \sum_{j=0}^n h_{\alpha} \biggl(\frac{j}{n}\biggr) \binom{n}{j} q^j (1-q)^{n-j} \\
	&= \sum_{j=0}^n \min\biggl\{\frac{j}{n},\alpha \biggr\} \binom{n}{j} q^j (1-q)^{n-j} \\
	&= 0 + q(1-q)^{n-1} + \sum_{j=2}^n \alpha \binom{n}{j} q^j (1-q)^{n-j} \\
	&= q(1-q)^{n-1} + \alpha \Bigl[ 1 - (1-q)^n - nq(1-q)^{n-1} \Bigr] \\
	&= \alpha \Bigl[ \frac{q}{\alpha}(1-q)^{n-1} + 1 - (1-q)(1-q)^{n-1} - nq(1-q)^{n-1} \Bigr] \\
	&= \alpha \Bigl[ 1 + (1-q)^{n-1} \Bigl( \frac{q}{\alpha} - 1 + q - nq \Bigr) \Bigr] \\
	&= \alpha \Bigl[ 1 - (1-q)^{n-1} \Bigl( 1 + (n-1)q - \frac{q}{\alpha} \Bigr) \Bigr] ,
\end{align*}
and the claim that $\frac{q}{\Bern_n(h_\alpha; q)} \leq s_{\alpha}(q)$ is equivalent to
\begin{align*}
	\frac{q}{\alpha \bigl[1-(1-q)^{n-1}\bigl(1+(n-1)q-\frac{q}{\alpha}\bigr)\bigr]} &\leq 1-q + \frac{q}{\alpha} , \\
	\frac{q}{\bigl((1-q)\alpha+q\bigr) \bigl[1-(1-q)^{n-1}\bigl(1+(n-1)q-\frac{q}{\alpha}\bigr)\bigr]} &\leq 1 , \\
	\bigl((1-q)\alpha+q\bigr) \Bigl[1-(1-q)^{n-1} \Bigl(1+(n-1)q-\frac{q}{\alpha}\Bigr)\Bigr] - q &\geq 0 , \\
	(1-q)\alpha - \bigl((1-q)\alpha+q\bigr) \Bigl[(1-q)^{n-1} \Bigl(1+(n-1)q-\frac{q}{\alpha}\Bigr)\Bigr] &\geq 0 , \\
	\alpha \Bigl[ 1-q - \Bigl( 1-q+\frac{q}{a} \Bigr) (1-q)^{n-1} \Bigl(1+(n-1)q-\frac{q}{\alpha}\Bigr) \Bigr] &\geq 0 ,~\text{and} \\
	\alpha (1-q) \Bigl[ 1 - (1-q)^{n-2} \Bigl( 1-q+\frac{q}{a} \Bigr) \Bigl(1+(n-1)q-\frac{q}{\alpha}\Bigr) \Bigr] &\geq 0 .
\end{align*}
Since $\alpha(1-q)\geq 0$ and $n=\lceil 1/\alpha\rceil$, it suffices to show that
\begin{align*}
	(1-q)^{\lceil 1/\alpha \rceil -2} \Bigl(1 - q + \frac{q}{\alpha}  \Bigr) \biggl(1 - q - \frac{q}{\alpha} + q \biggl\lceil\frac{1}{\alpha}\biggr\rceil \biggr) \leq 1 .
\end{align*}
Note that $(1-q+\frac{q}{\alpha})\geq 0$ and $(1-q-\frac{q}{\alpha}+q\lceil\frac{1}{\alpha}\rceil)\geq 0$. Note further that $\alpha<1$ and thus $\lceil 1/\alpha\rceil \geq 2$, so that by \lemref{lem:bernoulli}, $(1-q)^{\lceil 1/\alpha\rceil-2}\leq (1+(\lceil 1/\alpha\rceil-2)q)^{-1}$. 
It thus suffices to show that
\begin{align*}
	\frac{1}{1+q\lceil \frac{1}{\alpha}\rceil -2q} \biggl(1 - q + \frac{q}{\alpha} \biggr)\biggl(1 - q - \frac{q}{\alpha} + q\biggl \lceil \frac{1}{\alpha} \biggr\rceil \biggr) \leq 1 .
\end{align*}
Since $1+q\lceil 1/\alpha\rceil-2q\geq 0$ this is equivalent to
\begingroup 
\begin{align*}
	1+q \bigg\lceil \frac{1}{\alpha}\bigg\rceil -2q &\geq \biggl(1 - q + \frac{q}{\alpha} \biggr)\biggl(1 - q - \frac{q}{\alpha} + q\biggl \lceil \frac{1}{\alpha} \biggr\rceil \biggr) \\
	&= \biggl[1 + q\biggl(\frac{1}{\alpha}-1\biggr) \biggr]\biggl[1 + q\biggl(\bigg\lceil\frac{1}{\alpha}\bigg\rceil - \frac{1}{\alpha} - 1 \biggr) \biggr] \\
	&= 1 + q\biggl(\frac{1}{\alpha}-1 + \bigg\lceil\frac{1}{\alpha}\bigg\rceil - \frac{1}{\alpha}-1\biggr) + q^2\biggl(\frac{1}{\alpha}-1\biggr)\biggl(\bigg\lceil\frac{1}{\alpha}\bigg\rceil -\frac{1}{\alpha} -1 \biggr) \\
	&= 1 + q\biggl(\bigg\lceil\frac{1}{\alpha}\bigg\rceil - 2\biggr) + q^2\biggl(\frac{1}{\alpha}-1\biggr)\biggl(\bigg\lceil\frac{1}{\alpha}\bigg\rceil -\frac{1}{\alpha} -1 \biggr)
\end{align*}
\endgroup
and to
\begin{align*}
	q^2\biggl(\frac{1}{\alpha}-1\biggr)\biggl(\bigg\lceil\frac{1}{\alpha}\bigg\rceil -\frac{1}{\alpha} -1 \biggr) \leq 0 ,
\end{align*}
which is satisfied because $q^2\geq 0$, $\frac{1}{\alpha}-1\geq 0$, and $\lceil 1/\alpha\rceil-1/\alpha-1\leq 0$.

We have shown that for all $\alpha\in(0,1)$ and $n=\lceil 1/\alpha \rceil$, and $q\in(0,1)$,
\begin{align*}  %
	\frac{q}{\Bern_n(h_\alpha; q)} &\leq s_{\alpha}(q) .
\end{align*}
To see that this holds in fact for all $n\geq\lceil 1/\alpha \rceil$, it suffices to show that for all $q\in[0,1]$,
\[
	\frac{q}{\Bern_{n}(h_{\alpha}; q)} \leq \frac{q}{\Bern_{\lceil 1/\alpha \rceil}(h_{\alpha}; q)} .
\]
The latter follows from concavity of $h_{\alpha}$ for all $\alpha \in (0,1]$ and from monotonicity of the approximation of a concave function $h:[0,1]\to\R$ by Bernstein polynomials in the sense that for all $q\in [0,1]$, $\Bern_n(h;q)\leq\Bern_{n+1}(h;q)\leq h(q)$~\citep[\S~7]{Phil2003}. 
\end{proof}

\subsection{Proof of \thmref{thm:irregular-lower}}
\label{app:lower}

\irregularLower*

\begin{proof}
Let $\epsilon>0$ and consider the discrete distribution with support $\{1,\frac{kn}{\epsilon}\}$, $\Pr[v_i=1]=1-\frac{\epsilon}{n^2}$, and $\Pr[v_i=\frac{kn}{\epsilon}]=\frac{\epsilon}{n^2}$. Let $X\sim\mathsf{Bin}(n,\epsilon/n^2)$ and $Y\sim\mathsf{Bin}(n-k,\epsilon/n^2)$. Clearly the optimal static price~$p$ must be either $1$ or $\frac{kn}{\epsilon}$, and $\Rsnk(1)=k$ and $$\Rsnk\biggl(\frac{kn}{\epsilon}\biggr)=\E[\min\{X,k\}]\frac{kn}{\epsilon}\leq \E[X]\frac{kn}{\epsilon} = \frac{\epsilon}{n} \cdot \frac{kn}{\epsilon}=k.$$ On the other hand, for dynamic prices $\vec{p}$ where $p_i=\frac{kn}{\epsilon}$ if $i\leq n-k$ and $p_i=1$ otherwise,
\begingroup 
\begin{align*}
	\lim_{\epsilon\rightarrow 0}\Rdnk(\vec{p}) 
	&= \lim_{\epsilon\rightarrow 0} \E[\min\{Y,k\}]\biggl(\frac{kn}{\epsilon}-1\biggr)+k \\
	&= \lim_{\epsilon\rightarrow 0} \biggl( \bigl(\P[Y<k]\cdot\E[Y] + \P[Y\geq k]\cdot k\bigr)\biggl(\frac{kn}{\epsilon}-1\biggr)\biggr) + k \\
	&\geq \lim_{\epsilon\rightarrow 0} \biggl(1-\Bigl(\frac{\epsilon}{n^2}\Bigr)^k\biggr) \biggl( \E[Y] \biggl(\frac{kn}{\epsilon}-1\biggr) + k \biggr) \\
	&= \lim_{\epsilon\rightarrow 0} \biggl(1-\Bigl(\frac{\epsilon}{n^2}\Bigr)^k\biggr) \biggl( (n-k)\frac{\epsilon}{n^2} \biggl(\frac{kn}{\epsilon}-1\biggr) + k \biggr) = (n-k)\frac{k}{n}+k=2k-\frac{k^2}{n}.
\end{align*}
\endgroup
Choosing $\epsilon$ small enough shows the claim. 
\end{proof}

\end{document}